%% file: main.tex
\documentclass[11pt,dvipsnames]{article}

\usepackage[pagebackref=true]{hyperref}
\input{shared_headers}
\input{arxiv_headers}

\title{Does block size matter in randomized block\\ Krylov low-rank approximation?}

\author{\normalsize 
Tyler Chen\thanks{New York University (\email{tyler.chen@nyu.edu}, \email{cmusco@nyu.edu})} 
~~~~Ethan N.\ Epperly\thanks{California Institute of Technology (\email{eepperly@berkeley.edu}, \email{ram900@berkeley.edu})} 
~~~~Raphael A.\ Meyer\footnotemark[2] 
~~~~Christopher Musco\footnotemark[1] 
~~~~Akash Rao\thanks{Washington State University (\email{akashgopinath.rao@wsu.edu})}
}
\date{\today}

\begin{document}

\maketitle

\begin{abstract}
\input{abstract}
\end{abstract}

\input{content}

\section*{Acknowledgments}

\input{acknowledgements}

\bibliographystyle{halpha}
\bibliography{refs}
\appendix

\input{appendix}

\end{document}

%% file: shared_headers.tex
\usepackage{amsfonts,amssymb,amsmath}
\usepackage{bbm}
\usepackage{booktabs}
\usepackage[labelfont=bf]{caption}
\usepackage{color}
\usepackage{doi}
\usepackage{enumitem}
\usepackage{float}
\usepackage{graphicx}
\usepackage[utf8]{inputenc}
\usepackage{mathtools}
\usepackage{nicefrac}       % compact symbols for 1/2, etc.
\usepackage{subcaption}
\usepackage{xspace}
\usepackage{url}
\usepackage{breakcites}
\usepackage{appendix}
\usepackage{algorithm}
\usepackage{algpseudocode}
\usepackage{algorithmicx}

\usepackage{mathrsfs}
\usepackage{soul}
\usepackage{mleftright}
\usepackage[most]{tcolorbox}

\usepackage{crossreftools}
\usepackage{contour}
\usepackage[normalem]{ulem}

\contourlength{0.8pt}

\renewcommand{\Pr}{\operatorname{Pr}}
\newcommand{\defeq}[0]{\ensuremath{\;{\vcentcolon=}\;}\xspace}

\newcommand{\bv}[1]{\mathbf{#1}}

\let\norm\relax

\DeclarePairedDelimiter\norm{\lVert}{\rVert}

\DeclareMathOperator{\poly}{poly}

\DeclareMathOperator*{\Var}{Var}

\DeclareMathOperator{\diag}{diag}

\DeclareMathOperator{\rank}{rank}

\DeclareMathOperator{\range}{range}

\newcommand{\eps}[0]{\ensuremath{\varepsilon}}
\let\epsilon\eps

\newcommand{\mat}[1]{\mathbf{#1}} % Matrix
\renewcommand{\vec}[1]{\boldsymbol{\mathrm{#1}}} % Vector
\newcommand{\vecalt}[1]{\boldsymbol{#1}} % Vector, but better for greek letters
\newcommand{\variableth}[1]{$#1$th}
\newcommand{\rev}[1]{#1}

\newcommand{\bmat}[1]{\begin{bmatrix} #1 \end{bmatrix}} % [   Bracket Matrix   ]
\renewcommand*{\backref}[1]{}
\renewcommand*{\backrefalt}[4]{%
  \ifcase #1 %
  (No citations.)% use this if no citations
  \or
  (Cited on page #2.)% use this if only one citation
  \else
  (Cited on pages #2.)% use this if multiple citations
  \fi
}

\DeclareMathOperator{\orth}{orth}

\newcommand{\order}{\mathcal{O}}
\newcommand{\orderish}{\widetilde{\mathcal{O}}}

\renewcommand{\hat}[1]{\widehat{#1}}
\renewcommand{\tilde}[1]{\widetilde{#1}}

\makeatletter
\newcommand*{\tp}{%
  {\mathpalette\@tp{}}%
}
\newcommand*{\@tp}[2]{%
  \raisebox{\depth}{$\m@th#1\intercal$}%
}
\makeatother

\newcommand{\mA}{\ensuremath{\mat{A}}\xspace}
\newcommand{\mB}{\ensuremath{\mat{B}}\xspace}
\newcommand{\mC}{\ensuremath{\mat{C}}\xspace}
\newcommand{\mD}{\ensuremath{\mat{D}}\xspace}

\newcommand{\mG}{\ensuremath{\mat{G}}\xspace}
\newcommand{\mH}{\ensuremath{\mat{H}}\xspace}
\newcommand{\mI}{\ensuremath{\mat{I}}\xspace}

\newcommand{\mK}{\ensuremath{\mat{K}}\xspace}

\newcommand{\mM}{\ensuremath{\mat{M}}\xspace}

\newcommand{\mQ}{\ensuremath{\mat{Q}}\xspace}

\newcommand{\mU}{\ensuremath{\mat{U}}\xspace}
\newcommand{\mV}{\ensuremath{\mat{V}}\xspace}
\newcommand{\mW}{\ensuremath{\mat{W}}\xspace}

\newcommand{\mZ}{\ensuremath{\mat{Z}}\xspace}

\newcommand{\mGamma}{\ensuremath{\mat{\Gamma}}\xspace}

\newcommand{\vc}{\ensuremath{\vec{c}}\xspace}

\newcommand{\vg}{\ensuremath{\vec{g}}\xspace}

\newcommand{\vq}{\ensuremath{\vec{q}}\xspace}

\newcommand{\vu}{\ensuremath{\vec{u}}\xspace}

\newcommand{\vx}{\ensuremath{\vec{x}}\xspace}
\newcommand{\vy}{\ensuremath{\vec{y}}\xspace}
\newcommand{\vz}{\ensuremath{\vec{z}}\xspace}

\newcommand{\valpha}{\ensuremath{\vecalt{\alpha}}\xspace}
\newcommand{\vbeta}{\ensuremath{\vecalt{\beta}}\xspace}

\newcommand{\cK}{\ensuremath{{\mathcal K}}\xspace}

\newcommand{\cN}{\ensuremath{{\mathcal N}}\xspace}

\newcommand{\cU}{\ensuremath{{\mathcal U}}\xspace}
\newcommand{\cV}{\ensuremath{{\mathcal V}}\xspace}

\newcommand{\bbR}{\ensuremath{{\mathbb R}}\xspace}

\newcommand{\llbracket}{[\![}
\newcommand{\rrbracket}{]\!]}

\newcommand{\fro}{\mathsf{F}}

\newcommand{\vvvert}[1]{#1|\!#1|\!#1|}
\newcommand{\nnnorm}[2][]{\vvvert{#1} #2 \vvvert{#1}}

\newcommand{\gap}{\mathrm{gap}}
\newcommand{\mingap}[1]{\Delta_{#1}}
\newcommand{\cond}[1]{\kappa_{#1}}

\newcommand{\const}{\mathrm{C}}
\newcommand{\doubleblindreview}[1]{\textbf{[redacted for double-blind peer review]}}

%% file: arxiv_headers.tex
\usepackage[margin=1in,footskip=0.25in]{geometry}
\usepackage{amsthm}
\usepackage{thmtools}

\definecolor{c0}{HTML}{ffb000}
\definecolor{c1}{HTML}{fe6100}
\definecolor{c2}{HTML}{dc267f}
\definecolor{c3}{HTML}{785ef0}
\definecolor{c4}{HTML}{648fff}
\usepackage[capitalize, nameinlink, noabbrev]{cleveref}

\hypersetup{
    colorlinks=true,      % false: boxed links; true: colored links
    linkcolor=c2,       % color of internal links (change box color with linkbordercolor)
    citecolor=c4,      % color of links to bibliography
    filecolor=c4,    % color of file links
    urlcolor=c4         % color of external links
}
\newcommand*{\email}[1]{\href{mailto:#1}{\nolinkurl{#1}} } 

\declaretheorem[name=Theorem]{theorem}
\numberwithin{theorem}{section}

\declaretheorem[name=Proposition,numberlike=theorem]{proposition}

\declaretheorem[name=Observation,numberlike=theorem]{observation}

\declaretheorem[name=Lemma,numberlike=theorem]{lemma}
\declaretheorem[name=Corollary,numberlike=theorem]{corollary}
\declaretheorem[name=Imported Theorem,numberlike=theorem]{importedtheorem}
\declaretheorem[name=Imported Lemma,numberlike=theorem]{importedlemma}

\theoremstyle{remark}
\declaretheorem[name=Remark,qed={\lower-0.3ex\hbox{$\diamond$}},numberlike=theorem]{remark}

\theoremstyle{definition}
\declaretheorem[name=Definition,numberlike=theorem]{definition}
\usepackage{mdframed}
\newmdtheoremenv{question}{Question}
\numberwithin{equation}{section}
\declaretheorem[name=Problem,numberlike=theorem]{problem}

\usepackage{titlesec}

\titleformat{\subsection}[runin]
  {\normalfont\normalsize\scshape\bfseries}{\thesubsection}{0.5em}{}[\textbf{.}]

\titleformat{\subsubsection}[runin]
  {\normalfont\normalsize\scshape}{\textit{\thesubsubsection}}{0.5em}{}[.]

\usepackage{charter}
\DeclareMathAlphabet{\mathcal}{OMS}{zplm}{m}{n}

\makeatletter
\def\hlinewd#1{%
	\noalign{\ifnum0=`}\fi\hrule \@height #1 \futurelet
	\reserved@a\@xhline}
\makeatother

\crefname{equation}{}{}
\crefname{section}{Section}{Sections}
\crefname{appendix}{Appendix}{Appendices}
\crefname{property}{property}{properties}


\makeatletter
\def\th@plain{%
  \thm@notefont{}% same as heading font
  \itshape % body font
}
\def\th@definition{%
  \thm@notefont{}% same as heading font
  \normalfont % body font
}
\makeatother

\newcommand{\repThmName}{REP_THM_NAME}
\declaretheoremstyle[bodyfont=\normalfont\itshape,notefont=\normalfont\bfseries, notebraces={\repThmName}{Restated}]{repthmstyle}
\declaretheorem[style=repthmstyle, name=\hspace{-0.25em}, numbered=no]{reptheoremAbstract}

\newenvironment{reptheorem}[1]{%
    \renewcommand{\repThmName}{\Cref{#1}}
    \begin{reptheoremAbstract}[~]
} {%
    \end{reptheoremAbstract}
    \renewcommand{\repThmName}{REP_THM_NAME}
}%

%% file: abstract.tex
We study the problem of computing a rank‑$k$ approximation of a matrix using randomized block Krylov iteration.
Prior work has shown that, for block size $b = 1$ or $b = k$, a $(1 + \varepsilon)$-factor approximation to the best rank-$k$ approximation can be obtained after $\tilde\order(k/\sqrt{\varepsilon})$ matrix-vector products with the target matrix. On the other hand, when $b$ is between $1$ and $k$, the best known bound on the number of matrix-vector products scales with $b(k-b)$, which could be as large as $\order(k^2)$. 
Nevertheless, in practice, the performance of block Krylov methods is often optimized by choosing a block size $1 \ll b \ll k$.
We \rev{address} this theory-practice gap by proving that randomized block Krylov iteration produces a $(1 + \varepsilon)$-factor approximate rank-$k$ approximation using $\tilde\order(k/\sqrt{\varepsilon})$ matrix-vector products for \emph{any} block size $1\le b\le k$. 
Our analysis relies on new bounds for the minimum singular value of a random block Krylov matrix, which may be of independent interest. Similar bounds are central to recent breakthroughs on faster algorithms for sparse linear systems [Peng \& Vempala, SODA 2021; Nie, STOC 2022].

%% file: content.tex
\section{Introduction}

Krylov subspace methods are among the most widely studied algorithms in numerical analysis and scientific computing \cite{LS13}. They have been recognized—alongside the simplex method and fast Fourier transform—as one of ten algorithms with “greatest influence on the development and practice of science and engineering in the 20th century” \cite{D+2000}.
Nevertheless, basic questions remain even for these fundamental methods.
This paper is concerned with one such question.

In particular, we focus on the use of Krylov methods for the \emph{low-rank approximation problem:}
\begin{problem}[Low-rank approximation] \label{problem:LRA}
Given $\mA\in\bbR^{n\times d}$ and an accuracy target $\varepsilon>0$, find a rank-$k$ matrix $\hat\mA$ such that:
\[
% \text{(a)}~
\nnnorm{ \mA - \hat\mA }
\leq (1+\varepsilon) \nnnorm{ \mA - \llbracket \mA \rrbracket_k } \quad\text{and}\quad
% \text{(b)}~
| \norm{\mA \vq_i}_2^2 - \sigma_i^2 | \leq \varepsilon \sigma_{k\rev{+1}}^2  ~~\text{for all } i=1,\ldots,k. 
\]
Here, $\nnnorm{\cdot}$ is either the Frobenius or spectral norm, $\llbracket \mA \rrbracket_k$ is an optimal rank-$k$ approximation to $\mA$, $\vq_i$ is the \variableth{i} right singular vector of $\hat\mA$, and $\sigma_i$ is the \variableth{i} largest singular value of $\mA$.
\end{problem}
\noindent \Cref{problem:LRA} asks for an approximation $\hat\mA$ with (a) approximation error competitive with the best possible rank-$k$ approximation $\llbracket \mA \rrbracket_k$ and (b) with singular vectors that capture the ``high-energy'' directions in \mA.  
Algorithms for obtaining these guarantees have been widely studied in both theoretical computer science \cite{Sar06,CW09,Woo14,DM16,CW17} and computational mathematics \cite{RST10,HMT11,GM13,Gu15,DI19,TW23}.
It is common to consider just condition (a).
\rev{The techniques used in this paper immediately lead to both guarantees.}

\emph{Randomized block Krylov iteration} (RBKI) is widely considered one of the most effective algorithms for low-rank approximation \cite{MM15,TW23}. It offers the best known theoretical guarantees for \Cref{problem:LRA} \cite{BCW22a,MMM24}, nearly matches lower bounds in natural models of computation \cite{SER18,BN23a}, and is successful in practice \cite{M+23,TW23}.
In order to approximate a general matrix \(\mA\in\bbR^{n \times d}\), we first generate a starting block matrix \(\mG\in\bbR^{n \times b}\) with \emph{block size} \(b\).
Often, \mG is chosen to be a standard Gaussian matrix.
Then, we construct a matrix \(\mZ\in\bbR^{d \times qb}\) with orthonormal columns whose span is the \emph{block Krylov subspace}
\begin{align}
\label{eq:krylov_def}
\cK_{q}(\mM,\mG) \defeq \operatorname{range}\left(
 \begin{bmatrix}
    \mG
    &
    \mM\mG
    &
    \mM^2\mG
    &\cdots 
    &
    \mM^{q-1}\mG
 \end{bmatrix}\right)
 ,\quad \text{ where } \mM \defeq \mA\mA^\tp.
\end{align}
We refer to the parameter \(q\) as the \emph{iteration count} of the Krylov subspace.
Finally, we return \(\mZ\llbracket\mZ^\tp\mA\rrbracket_k\), which is the best Frobenius norm rank-$k$ approximation to \mA in the Krylov subspace.
This meta-procedure is described in \Cref{alg:block_krylov_LRA}.
Practical implementations differ in exactly how the orthonormalization step in line \ref{line:build_Kq} is performed and how $\hat\mA$ is computed (e.g., it is usually desirable to return the matrix in factored form). We refer the reader to \cite{S97,TW23} for more details on implementation. In general, however, the cost of \Cref{alg:block_krylov_LRA} is dominated by the cost of constructing an orthonormal basis for \(\cK_q(\mA\mA^\tp,\mG)\).
Doing so, for instance via the block Lanczos method \cite{golub1977block,golub2013matrix}, expends at least
\[
 \order\left(bq \cdot \operatorname{mv}(\bv{A}) + bq^2n \right)
 ~\text{arithmetic operations,}
\]
where $\operatorname{mv}(\mA)$ is the cost of performing a matrix-vector product with $\mA$ and $\mA^\tp$.
For dense $\mA$, $\operatorname{mv}(\mA) = \order(nd)$, and the asymptotic cost is typically dominated by the first term. 

\begin{algorithm}[t]
\caption{Randomized block Krylov iteration (RBKI) for low-rank approximation}\label{alg:block_krylov_LRA}
\begin{algorithmic}[1]
\Require Target matrix $\mA\in\bbR^{n\times d}$, rank parameter $k$, block size $b$, iteration count $q$
\State Sample $n\times b$ Gaussian matrix $\mG$
\State Build $\mZ = \orth(\cK_{q}(\mA\mA^\tp,\mG))$ 
\label{line:build_Kq}
\State Compute $\llbracket \mZ^\tp\mA \rrbracket_k$ 
\Ensure $\hat\mA \defeq \mZ \llbracket\mZ^\tp \mA \rrbracket_k$
\end{algorithmic}
\end{algorithm}

\subsection{The intermediate block size mystery}
\label{sec:matvec-mystery}
Since matrix-vector products typically dominate the runtime cost, several papers use the number of matrix-vector products as a proxy for the runtime of \cref{alg:block_krylov_LRA}.
Existing analyses show that, if we choose block size $b=1$ \cite{MMM24} or $b = k$ \cite{MM15,TW23}, then $bq = \tilde\order(k / \sqrt{\varepsilon})$ matrix-vector products suffice to solve \Cref{problem:LRA}.
Lower bounds show that this complexity is optimal for rank-1 approximation when the norm $\nnnorm{\cdot}$ in \Cref{problem:LRA} is the spectral norm \cite{SER18,BN23a}.
For the Frobenius norm, RBKI achieves an improved complexity of $bq = \tilde\order(k/\varepsilon^{1/3})$ \cite{BCW22a,MMM24}, which is also essentially optimal \cite{BN23a}. 

\begin{figure}
    \centering
    \includegraphics[scale=.8]{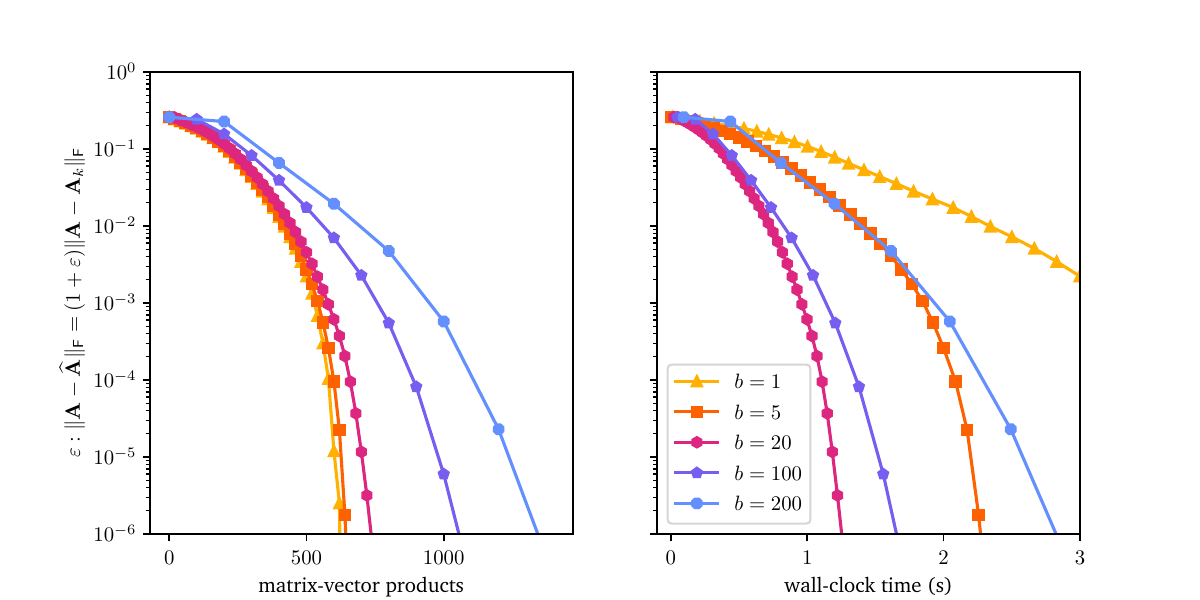}
    \caption{Accuracy versus cost to compute a rank $k=200$ approximation for a $2000\times 2000$ dense matrix $\mA$ using \Cref{alg:block_krylov_LRA}.
    \rev{See \cref{sec:more-figs} for details.}
    While smaller block sizes perform best in terms of matrix-vector products, the fastest choice in terms of wall-clock time is an \emph{intermediate} block size, $b=20$.
    Prior to our work, theoretical guarantees for intermediate block sizes  lagged behind those for $b=1$ or $k$.
    }
    \label{fig:intro}
\end{figure}

However, while the \rev{asymptotic analysis} of RBKI is basically complete for two extremes, $b=1$ and $b=k$, our understanding of ``intermediate'' choices for the block size, $1 < b < k$, remains incomplete.
This is despite the fact that, \textbf{in practice, $\boldsymbol b$ is most commonly chosen to lie somewhere \emph{between} $\boldsymbol{1}$ and $\boldsymbol{k}$}.
We will discuss why this is the case shortly. For general $b$, the best existing results for RBKI were developed in the work of Meyer, Musco, and Musco \cite{MMM24}, which shows:
\begin{importedtheorem}[RBKI with intermediate block size, \protect{\cite[Thm.~4.5]{MMM24}}] \label{impthm:mmm24}
    Fix a rank $k$ and block size $1\le b \leq k$.
    \Cref{alg:block_krylov_LRA} solves \Cref{problem:LRA} with probability at least \(0.99\) with total matrix-vector complexity $bq = \tilde \order(b(k-b + 1)/\sqrt{\varepsilon})$.
\end{importedtheorem}
\noindent Here, and going forward, $\tilde O$ suppresses logarithmic factors in the dimensions of $\mA$, the accuracy $\eps$, the rank-\(k\) condition number, and the inverse-gaps between the singular values.
(The last of these can be eliminated by an initial random perturbation of the input; see \cref{sec:smoothed-analysis}).

For small block sizes $b = \order(1)$ or almost-$k$ block sizes $b = k - \order(1)$, the above theorem recovers the expected $\tilde \order(k/\sqrt{\varepsilon})$ matrix-vector complexity for RBKI.
However, for intermediate block sizes the result deteriorates, and the  complexity can be as bad as $\tilde \order(k^2/\sqrt{\varepsilon})$ for block size $b = k/2$.

Importantly, this deterioration does not seem to be consistent with experimental results. Indeed, as mentioned above, the speed of RBKI in practice is often optimized by choosing a block size $b$ {between} $1$ and $k$. This phenomenon can be explained by the competition of two effects:
\begin{itemize}
    \item \textbf{Smaller is better.}
    While the worst-case matrix-vector product complexity for block sizes $b=1$ and $b=k$ match, the convergence of Krylov methods on a given instance \mA is highly dependent on $\mA$'s spectrum, so ``better-than-worst-case'' performance is often observed. Theoretical evidence suggests that\rev{, in the absence of small singular value gaps,} choosing $b=1$ typically requires fewer matrix-vector products to reach a desired level of accuracy for non-worst-case instances; see  \cite{MMM24} or \rev{\cref{rem:small-block-benefits}} for more \rev{discussion}. 
    Empirical support for this claim is provided in \Cref{fig:intro}: Smaller $b$ leads to fewer matrix-vector products.
    \item \textbf{Bigger is better.} Due to parallelism, caching, and dedicated hardware accelerators in modern computing systems, it is typically much faster to compute multiple matrix-vector products all at once (i.e., grouped into a matrix-matrix product) instead of one after the other. When run with block size $b$, \Cref{alg:block_krylov_LRA} groups matrix-vector products into blocks of size $b$, so for a fixed number of total matrix-vector products, we expect the method to run faster if $b$ is larger. 
    The speedups here can be large: On an author's laptop computer, multiplying $\mA \in \bbR^{2000\times 2000}$ with $\mG \in \bbR^{2000\times 200}$ was $9\times$ faster than multiplying $\mA$ with a single vector, $200$ times in sequence.
    The benefits of blocking into matrix-matrix products can be even more pronounced when $\mA$ is too large to store in fast memory, in which case the cost of reading $\mA$ from slow memory is a bottleneck, and using a larger block size $b\gg 1$ can be essentially ``free'' \cite{DTLKGYAHA17,CHLZ25}.
\end{itemize}
In practice, the ``smaller is better'' and ``bigger is better'' effects compete with each other, and the fastest runtime is usually obtained by an intermediate block size that is much larger than $1$ but much smaller than the rank $k$.
An example of this effect is shown in \Cref{fig:intro}.

Another, more prosaic, reason why intermediate block sizes are used is \emph{adaptivity}.
Often, practitioners are interested in computing a low-rank approximation $\hat\mA$ to meet an error tolerance $\nnnorm{\mA - \hat\mA} \le \tau$, 
and the rank $k$ is chosen adaptively at runtime to meet this tolerance \cite[\S6.3]{TW23}.
Since $k$ is not known in advance, it cannot be used to inform the choice of block size $b$.

\subsection{Our main result}
This work is concerned with \rev{addressing} the theory-practice gap for RBKI with intermediate block sizes.
We do so by providing an analysis of RBKI that achieves a \emph{linear dependence on $k$} for \emph{any} block size $1 \le b \le k$, improving on the best existing result, \cref{impthm:mmm24}, by up to a quadratic factor. 
Specifically, we prove the following:
\begin{theorem}[RBKI with any block size, gap-independent bounds, informal version]\label{thm:LRA}
Fix a rank $k$ and block size $1\le b \le k$.
\Cref{alg:block_krylov_LRA} solves \Cref{problem:LRA} with probability at least \(0.99\) with total matrix-vector complexity $bq = \tilde \order(k/\sqrt{\varepsilon})$.
\end{theorem}
See \cref{sec:proof-main} for a full version of this result and its proof.
Similar to prior work \cite{MM15,TW23,MMM24}, we also prove a version of \Cref{thm:LRA} that can be sharper when \mA has larger gaps between singular values; see \cref{sec:gap-dependent-bound}.
Overall, our results provide a rigorous justification for the use of Krylov iteration with {any} block size between $1 \leq b \leq k$, suggesting that practitioners can confidently select b to optimize efficiency on their hardware without compromising theoretical guarantees.

\subsection{Conditioning of Krylov matrices}
In the proof of \cref{thm:LRA}, it is simple to reduce to the case where \(k\) is divisible by \(b\).
Using this assumption, the first key observation in our proof is that the analysis of RBKI (\cref{alg:block_krylov_LRA}) can be reduced to understanding the minimum singular value of a random square block Krylov matrix
\begin{equation} \label{eq:krylov-mat-intro}
    \mK \coloneqq \begin{bmatrix}
    \mH & \mC \mH & \cdots & \mC^{t-1}\mH
\end{bmatrix} \in \bbR^{k\times k},
\end{equation}
where \(t=k/b\), $\mH \in \cN(0,1)^{k \times b}$ is a Gaussian starting block, and $\mC \in \bbR^{k\times k}$ is a real symmetric  positive semi-definite matrix with spectral norm \(\norm\mC_2=1\).
The columns of \mK are highly redundant; \rev{as in the power method, $\mC^{i}\mH$ will converge to the top $b$ eigenvectors of $\mH$ as $i$ increases.
Therefore,} we should expect that \mK is very nearly singular.
Fortunately, our analysis will only require proving that $\sigma_{\rm min}(\mK)$ is at worst exponentially small in $t$, i.e., \(\log(1/\sigma_{\rm min}(\mK)) \leq \orderish(t)\).

The problem of understanding the conditioning of a random block Krylov matrix is not new: several results in computational linear algebra require proving that $\mK$ is non-singular with high probability, i.e., that $\sigma_{\min}(\mK) > 0$ \cite{K95,EberlyGiesbrechtGiorgi:2006}. Even such a weak bound is non-trivial, but can be established via an application of the Schwartz--Zippel lemma (see \cref{sec:SZ} for details). 

Proving a quantitative lower bound is more difficult. 
Indeed, such a bound is central in Peng and Vempala's breakthrough work on solving sparse linear systems faster than matrix multiplication time \cite{PV21}
(though, they study a more general setting where $\mH$ can be a sparse random matrix).
\cite{PV21} inspired follow-up work of Nie \cite{Nie22}, who developed new matrix anti-concentration inequalities that can be used to obtain sharper bounds on $\sigma_{\min}(\mK)$.
However, even Nie's improvements can only show that \(\log(1/\sigma_{\rm min}(\mK)) \leq \orderish(t)\) when the block size is at least \(b\geq\tilde\Omega(\sqrt k)\).

In this paper, we simplify and extend Peng, Vempala, and Nie's arguments to hold for all $b \leq k$.
Our proof is relatively short, so we refer the reader to our main technical section, \cref{sec:krylov_min_singularvalue}, for details.
Overall, we obtain the following result:
\begin{theorem}[Random block Krylov matrices are not too ill-conditioned, informal version]\label{thm:krylov_sigmamin_bound_intro}
Let $\mC\in\bbR^{k\times k}$ be a symmetric positive semi-definite matrix with $\norm{\mC}_2 = 1$, let $\mH \sim \cN(0,1)^{k\times b}$ be a Gaussian starting block with $b \leq k$ columns, where $b$ divides $k$. Set $t= b/k$. For $\mK$ as defined in \eqref{eq:krylov-mat-intro}, we have
$\log(1/\sigma_{\rm min}(\mK))
\leq \tilde{\order}(t)$ with probability at least $0.99$.
\end{theorem}
\noindent 
The full version of \cref{thm:krylov_sigmamin_bound_intro} is stated as \cref{thm:krylov_sigmamin_bound} in \cref{sec:krylov_min_singularvalue}.
We hope that this result---and the techniques used in its proof---will be of interest beyond low-rank approximation.
In addition to studying Peng--Vempala-style algorithms, our techniques could be useful in studying the numerical stability of randomized Krylov subspace methods for linear systems \cite{NT24,BG25,YHBS25}.

\section{Preliminaries} \label{sec:notation}

Throughout, we work with a matrix $\mA \in \bbR^{n\times d}$ with singular values are $\sigma_1,\sigma_2,\ldots$.
It will also be helpful to work with the matrix $\mM = \mA\mA^\tp$, whose eigenvalues are denoted $\lambda_i \defeq \sigma_i^2$.
Throughout, $\norm{\,\cdot\,}_p$ denotes the $\ell_p$ norm of a vector or the operator $\ell_p\to\ell_p$ norm of a matrix. We use $\orth(\cV)$ to denote the operation that returns any set of orthonormal columns spanning a subspace $\cV$. We let 
$\cK_{q}(\mM,\mB)$ denote the degree $q-1$ Krylov subspace for matrix $\mM$ and starting block $\mB$, as in \cref{eq:krylov_def}.

\subsection{Gaps and asymptotic notation}

The analysis of RBKI depends on gaps between singular values.
To this end, we define the minimum relative singular value gap and rank-\(k\) condition number:
\begin{equation}
    \label{eqn:gap_def}
    \mingap{k} \coloneqq \min_{i=1,\ldots,k-1} \frac{\sigma_i^2 - \sigma_{i+1}^2}{\sigma_i^2} = \min_{i=1,\ldots,k-1} \frac{\lambda_i - \lambda_{i+1}}{\lambda_i}\quad \text{and}\quad \cond{k} \defeq \frac{\sigma_1^2}{\sigma_{k-1}^2}.
\end{equation}
The additive singular value gap can be bounded by the relative gap and condition number:
\begin{observation} [Additive to relative gaps]
\label{lem:additive_gap_to_Delta}
It holds that
\begin{equation*}
    \min_{i=1, \ldots, k-1} \big[\lambda_{i} - \lambda_{i+1}\big] = \min_{i=1, \ldots, k-1} \big[\sigma_i^2 - \sigma_{i+1}^2\big] \geq  \lambda_1 \cdot \frac{\mingap{k}}{ \cond{k}}.
\end{equation*}
\end{observation}
\noindent
See \cref{rem:gaps} for more discussion about additive and relative gaps in our analysis.
For this paper, the asymptotic notation, $\tilde \order$, suppress logarithmic factors in the dimensions of $\mA$, the accuracy $\eps$, the failure probability $\delta$, the condition number $\cond{k}$, and the inverse-minimum gap $1/\mingap{k}$.

\subsection{(Anti-)concentration of Gaussian variables}

We require a couple of concentration results.
First, we use the following very crude bound for the norm of a Gaussian vector:

\begin{importedlemma}[Norm concentration, simplified version of \protect{\cite[Lem.~1]{LM00}}]
\label{lem:chisq_concentration}
Let $\vg \sim \cN(0,1)^p$ be a vector with $p$, i.i.d. Gaussian entries.
For any $\xi\geq1$, $\Pr[ \norm{\vg}_2^2 \geq 5\xi p ] < \exp(-\xi p)$.
\end{importedlemma}

Second, we require a standard anti-concentration bound for the Gaussian distribution, which follows by noting that the pdf of the standard normal distribution is bounded by $1/2$. 
\begin{proposition}[Gaussian anti-concentration]
\label{lem:gaussian_anti_concentration}
    Let $g \sim \mathcal{N}(0,\sigma^2) $. For any $\alpha>0$, $\Pr[|g|<\alpha]< \frac{\alpha}{\sigma}$.
\end{proposition} 

\section{Randomized block Krylov with any block size}
\label{sec:overview}
This section \rev{contains the proof of} our main theorem, \Cref{thm:LRA}, assuming a bound on the minimum singular value of a randomized block Krylov matrix. That minimum singular value bound (our main technical contribution, \Cref{thm:krylov_sigmamin_bound}) is then proven in \Cref{sec:krylov_min_singularvalue}.

\subsection{Good starting blocks}
Our general approach to analyzing \Cref{alg:block_krylov_LRA} follows \cite{MMM24}, which builds on \cite{MM15}.
The key property we require for the analysis is a measure of \rev{the degree of overlap} between the stating block $\mB$ and the dominant eigenspace of $\mA\mA^\tp$.

\begin{definition}
[Good starting block]\label{def:kL_good}
Let $\mA\in\mathbb{R}^{n\times d}$ be a matrix with top $k$ left singular vectors  $\mU_k \in \bbR^{n\times k}$. For $\ell \geq k$, a matrix $\mB\in \bbR^{n\times \ell}$ is a $(k,L)$-good starting block for $\mA$ if $\norm{(\mU_k^\tp\mQ)^{-1}}_2^2\leq L$ for some matrix $\mQ\in\bbR^{n\times k}$ with orthonormal columns that lie in $\range(\mB)$.
\end{definition}

\cite{MM15} implicitly analyzes the quality of the low-rank approximation produced by RBKI on a starting block $\mB$ in terms of the parameters $k$ and $L$.
This approach is made explicit in \cite{MMM24}:

\begin{importedtheorem}[$(k,L)$-good implies RBKI works, \protect{\cite[App.~F]{MMM24}}]
\label{thm:kLgood_implies_LRA}
Let $\mB\in\bbR^{n\times \ell}$ for $\ell \geq k$, fix  $s>0$, and define $\mZ \defeq \orth(\cK_s(\mM,\mB))$.
If $\mB$ is $(k,L)$-good for $\mA$ and
$s = \order( \log( {n L }/{\varepsilon} ) / \sqrt{\varepsilon} )$, then $\hat\mA = \mZ \llbracket \mZ^\tp\mA\rrbracket_k$ solves \cref{problem:LRA}.
\end{importedtheorem}

In their original paper for large-block case $b = k$ \cite{MM15}, Musco and Musco analyze RBKI by showing that a Gaussian starting block $\mG \in \bbR^{n\times b}$ is $(k,\poly(n))$-good with high probability.
\Cref{thm:kLgood_implies_LRA} then implies RBKI solves \Cref{problem:LRA} with $\tilde \order(k/\sqrt{\varepsilon})$ matrix-vector products.

Observe that the \((k,L)\)-good property requires that \mB has at least \(k\) columns, preventing us from directly applying \cref{thm:kLgood_implies_LRA} to block sizes \(b < k\).
Instead, we use the \emph{simulated starting block argument} of \cite{MMM24}.
For ease of exposition, we will initially assume that $b$ evenly divides into $k$, and we define $t = k/b$.
The essence of the simulated starting block argument is that a $q$-iteration block Krylov subspace with starting block \mG coincides with a $s=(q-t+1)$-iteration block Krylov subspace associated with the \emph{simulated starting block}
\begin{equation}\label{eqn:simulated_starting_block}
\mB 
\defeq \begin{bmatrix}
\mG & \mM \mG & \cdots & \mM^{t-1} \mG
\end{bmatrix}\in\bbR^{n\times k}.
\end{equation}
Specifically, we have the following result:

\begin{observation}[Simulated starting block, \protect{\cite[Sec.~3.1]{MMM24}, \cite[Eq.~3.1]{CH23}}] \label{fact:shared_krylov}
Let $\mB$ be as in \Cref{eqn:simulated_starting_block}.
Then $\cK_{s+t-1}(\mM,\mG) = \cK_{s}(\mM,\mB)$.
\end{observation}

We see that to analyze block Krylov iteration with the starting block $\mG$ of size $b$, it suffices to analyze block Krylov iteration on the simulated starting block $\mB$ of size $k$.
Since \mB has size $k$, we are now able to utilize \cref{thm:kLgood_implies_LRA}.
To do so, we will establish bounds on the $(k,L)$-goodness of \mB that improve on \cite{MMM24} for intermediate block sizes $1\ll b\ll k$. 

\subsection{Reduction to the conditioning of a square Krylov matrix} \label{sec:reduction-to-square}

The first key observation in our analysis is that the $(k,L)$-goodness property (\Cref{def:kL_good}) can be reduced to bounding an expression involving the minimum singular value of $\mU_k^\tp \mB$.

\begin{lemma}[From minimum singular value to $(k,L)$-good]
    \label{lem:kl-good-to-min-sing-val-square}
    Let $\mA\in\mathbb{R}^{n\times d}$ be a matrix with top $k$ left singular vectors  $\mU_k \in \bbR^{n\times k}$.
    Suppose $bt=k$ and $\sigma_1 =1$.
    Then, with probability at least \(1 - \delta \), we have that the simulated starting block \( \mB\) defined in \Cref{eqn:simulated_starting_block} is \((k,L)\)-good for \mA with
    \[
        L = 
        \frac{5 \max\{kn,t\log(1/\delta)\}}{\sigma_{\rm min}^2(\mU_k^\tp\mB)}.
    \]
\end{lemma}

This observation was not made in \cite{MM15,MMM24}, who established $(k,L)$-goodness using a different approach.
\Cref{lem:kl-good-to-min-sing-val-square} allows us to reduce the analysis of RBKI with an intermediate block size to a random matrix theory question about understanding the minimum singular value of $\mU_k^\tp \mB$, a Krylov matrix.
Indeed, owing to the rotational invariance of Gaussian matrices and the rotational invariance of singular values, without loss of generality we can assume \mM is diagonal, in which case $\mU_k$ consists of the first $k$ columns of the identity matrix and $\mU_k^\tp \mB$ takes the form 
\begin{equation} \label{eq:UkB_is_block_Krylov}
    \mU_k^\tp \mB = \begin{bmatrix}
        \mH & \mC \mH & \cdots & \mC^{t-1}\mH
    \end{bmatrix} \in \bbR^{k\times k},
\end{equation}
where 
$\mC = \diag(\lambda_1,\ldots,\lambda_k)$ and $\mH \sim \cN(0,1)^{k\times b}$.
In particular, we recognize $\mU_k^\tp \mB$ as a \emph{square} random block Krylov matrix.
We will study the minimum singular value of such matrices in \cref{sec:krylov_min_singularvalue}.

\begin{proof}[Proof of \Cref{lem:kl-good-to-min-sing-val-square}]
Let $\mQ = \orth(\range(\mB))$.
If $\rank(\mB) < k$, the statement is vacously true with $L=\infty$. 
Therefore, we assume $\mB$ and $\mQ$ are rank-$k$ and hence injective.
Since \mQ has orthonormal columns, $\norm{\mQ\vx}_2 = \norm{\vx}_2$.
Therefore,
\[
L = \norm{(\mU_k^\tp \mQ)^{-1}}_2^2
= \max_{\vz\neq\vec{0}} \frac{\norm{(\mU_k^\tp \mQ)^{-1}\vz}_2^2}{\norm{\vz}_2^2}
= \max_{\vx\neq\vec{0}} \frac{\norm{\vx}_2^2}{\norm{(\mU_k^\tp \mQ)\vx}_2^2}
= \max_{\vx\neq\vec{0}} \frac{\norm{\mQ\vx}_2^2}{\norm{\mU_k^\tp (\mQ\vx)}_2^2}.
\]
Since $\range(\mQ) = \range(\mB)$, for any vector $\vx$, there is a vector $\vc$ so that $\mQ\vx = \mB \vc$.
Thus,
\[
\max_{\vx\neq\vec{0}} \frac{\norm{\mQ\vx}_2^2}{\norm{\mU_k^\tp (\mQ\vx)}_2^2}
= \max_{\vc\neq \bv{0}} \frac{\norm{\mB\vc}_2^2}{\norm{\mU_k^\tp\mB\vc}_2^2}
\leq \frac{\sigma_{\rm max}^2(\mB)}{\sigma_{\rm min}^2(\mU_k^\tp\mB)}.
\]
We directly bound the maximum singular value of \(\mB\).
Observe
\[
    \sigma_{\rm max}^2(\mB)
    \leq \norm{\mB}_\fro^2
    = \sum_{i=0}^{t-1} \norm{\mM^i \mG}_\fro^2
    \leq \sum_{i=0}^{t-1} \norm{\mM^i}_2^2\norm{\mG}_\fro^2
    \leq t \|\bv{G}\|_\fro^2.
\]
Here, we use that $\norm{\mM}_2 = \sigma_1^2 = 1$.
Since $\mG$ has $nb$ i.i.d. Gaussian entries,  \Cref{lem:chisq_concentration} implies that $\Pr[\|\mG\|_\fro^2 > 5\xi bn] < \exp(-\xi bn)$ for any $\xi \ge 1$.
Ergo,
\[
\Pr\left[ \sigma_{\rm max}(\mB) \geq 5\xi kn \right]
\le 
\Pr\left[ \sigma_{\rm max}(\mG) \geq 5\xi bn \right]
\leq \exp(-\xi bn). 
\]
Set $\xi = \max \{ 1, \log(1/\delta)/bn\}$ to complete the proof.
\end{proof}

\subsection{Proof of \cref{thm:LRA}} \label{sec:proof-main}

By \Cref{eq:UkB_is_block_Krylov}, we see that $\mU_k^\tp\mB$ is a square block Krylov matrix.
As such, we can bound its minimum singular value with \cref{thm:krylov_sigmamin_bound} (formal version of \cref{thm:krylov_sigmamin_bound_intro}), proven below in \cref{sec:krylov_min_singularvalue}.
Using this observation, analysis of RBKI is immediate.
To remove the assumption $bt = k$, we first make the following observation, which we prove in \cref{sec:proofs}.

\begin{lemma}[Smaller rank is easier]\label{lem:kpLgood_klgood}
    If $\mB$ is $(k',L)$-good for $\mA$, then it is $(k,L)$-good for any $k\leq k'$.
\end{lemma}

With this preparation, we now state and prove a formal version of \Cref{thm:LRA}:

\begin{reptheorem}{thm:LRA}
Let $t = \lceil k/b \rceil$ and define $k' \defeq bt$.
Then running \Cref{alg:block_krylov_LRA} with 
\[
q  = \order\left( \frac{t}{\sqrt{\varepsilon}} \log\left( \frac{\cond{k'}}{\mingap{k'}} \right)
+ \frac{1}{\sqrt{\varepsilon}} \log\left( \frac{n}{\delta \varepsilon} \right)\right)
\]
solves \Cref{problem:LRA} with probability at least $1-\delta$.
The total matrix-vector complexity is
\begin{equation*}
    bq  = \order\left( \frac{k}{\sqrt{\varepsilon}} \log\left( \frac{\cond{k'}}{\mingap{k'}} \right)
+ \frac{b}{\sqrt{\varepsilon}} \log\left( \frac{n}{\delta \varepsilon} \right)\right).
\end{equation*}
Here $\mingap{k'} := \min_{i=1, \ldots, k'-1} (\lambda_i - \lambda_{i+1} )/\lambda_i$ is the minimum (relative) gap between the eigenvalues $\lambda_i$ of $\mA\mA^\tp$ and $\cond{k'}:= \lambda_1 / \lambda_{k'-1}$ is the rank-$k$ condition number.
\end{reptheorem}

\begin{proof}
Since \cref{alg:block_krylov_LRA} and \cref{problem:LRA} are both scale-invariant, we are free to assume $\sigma_1 = 1$.
As noted in \cref{eq:UkB_is_block_Krylov}, $\mU_k^\tp\mB$ is a Krylov matrix. 
Hence, by \cref{thm:krylov_sigmamin_bound} and \cref{eq:UkB_is_block_Krylov}, we have that
\[
\Pr\left[
\sigma_{\rm min}\left( \mU_k^\tp\mB \right) \geq \const \cdot \frac{(\delta/2)^5}{k^{14}}\left(\frac{\mingap{k'}}{6\cond{k'}}\right)^{6(t-1)}\right] 
\geq 1-\frac{\delta}{2}.
\]
Combining this with \Cref{lem:kl-good-to-min-sing-val-square} with failure probability $\delta/2$ and union bounding, we conclude that, with probability at least $1-\delta$,
$\mB$ is $(k',L)$-good for 
\begin{equation*}
    L = \order \left( \frac{k^{29}n\log(1/\delta)}{\delta^{10}} \left(\frac{6\cond{k'}}{\mingap{k'}}\right)^{6(t-1)}\right).
\end{equation*}
By \Cref{lem:kpLgood_klgood}, \mB is also $(k,L)$-good.
The advertised conclusion follows by \Cref{fact:shared_krylov,thm:kLgood_implies_LRA}.
\end{proof}

\begin{remark}[Additive and relative gaps] \label{rem:gaps}
    We note  that \Cref{thm:LRA} (and related results \Cref{thm:krylov_sigmamin_bound,thm:LRA_gap}) hold with the same proofs with the ratios $\mingap{k}/\cond{k}$ replaced by the \emph{minimum additive gap} $\Delta^{\rm add}_k \coloneqq \min_{i = 1,\ldots,k-1} (\lambda_i - \lambda_{i+1}) / \lambda_1$.
    The additive gap $\Delta^{\rm add}_k$ is always larger than the ratio $\mingap{k}/\cond{k}$, in view of \Cref{lem:additive_gap_to_Delta}.
    We use the relative gaps $\mingap{k}$ and condition numbers $\cond{k}$ in this paper for consistency with past work \cite{MMM24}. 
    We conjecture that \Cref{thm:LRA,thm:LRA_gap} hold with $\mingap{k}$ in place of $\mingap{k}/\cond{k}$; see \cref{sec:conclusion} for some further discussion.
\end{remark}

\subsection{Gap-dependent bound} \label{sec:gap-dependent-bound}

Similar to classical analysis of block Krylov iteration \cite{Saa80} and prior works on \emph{randomized} block Krylov iteration \cite{MM15,TW23,MMM24}, we can obtain improved convergence guarantees when the matrix $\mA$ has a gap in its spectrum.
Specifically, define the gap from the \variableth{k} to \variableth{\ell} squared singular values as
\[
\gap_{{k\to \ell}}
\defeq 
\frac{\sigma_{k}^2 - \sigma_{\ell}^2}{\sigma_k^2}.
\]
Note that $\gap_{k\to \ell}$ represents the gap between \emph{individual} singular values, whereas the minimum gap $\Delta_k$ in \Cref{eqn:gap_def} represent the \emph{minimum} gap between all pairs of squared singular values.
To prove gap-dependent bounds, we rely on the following established result:

\begin{importedtheorem}[$(k,L)$-good implies RBKI works, gap-dependent, \protect{\cite[app.~F]{MMM24}}]
\label{thm:kLgood_implies_LRA_gap}
Let $\mB\in\bbR^{n\times b}$ and fix $s>0$. 
Define
$\mZ \defeq \orth(\cK_s(\mM,\mB))$.
Then $\hat\mA = \mZ \llbracket \mZ^\tp\mA\rrbracket_k$ solves \Cref{problem:LRA} provided $\mB$ is $(\ell,L)$-good for $\mA$ and
$s = \order( \log( {n L }/{\varepsilon} ) / \sqrt{\gap_{k\to \ell}})$
\end{importedtheorem}

This result shows that for a constant-size gap, \(\gap_{k\to \ell} = \Theta(1)\), RBKI converges in just \(q = \order(\log(nL/\eps))\) iterations; the dependence \(\eps\) is only logarithmic.
Combining this result with \Cref{thm:krylov_sigmamin_bound,lem:kl-good-to-min-sing-val-square,fact:shared_krylov} immediately yields the following result:

\begin{theorem}[RBKI with any block size, gap-dependent bounds]
\label{thm:LRA_gap}
Fix block size $1\le b \le k$ and choose any $\ell>k$. Define $t \defeq \lceil \ell/b \rceil$ and define $\ell' \defeq bt$.
Then running \Cref{alg:block_krylov_LRA} for 
\[
q  = \order\left( \frac{t}{\sqrt{\gap_{k\to\ell}}} \log\left( \frac{\cond{\ell'}}{\mingap{\ell'}} \right)
+ \frac{1}{\sqrt{\gap_{k\to\ell}}} \log\left( \frac{n}{\delta \varepsilon} \right)\right) \text{\,\, iterations}
\]
solves \Cref{problem:LRA} with probability at least $1-\delta$.
The total matrix-vector complexity is
\begin{equation*}
    bq  = \order\left( \frac{\ell}{\sqrt{\gap_{k\to \ell}}} \log\left( \frac{\cond{\ell'}}{\mingap{\ell'}} \right)
+ \frac{b}{\sqrt{\gap_{k\to \ell}}} \log\left( \frac{n}{\delta \varepsilon} \right)\right).
\end{equation*}
Here $\mingap{k'} := \min_{i=1, \ldots, k'-1} (\lambda_i - \lambda_{i+1} )/\lambda_i$ is the minimum (relative) gap between the eigenvalues $\lambda_i$ of $\mA\mA^\tp$ and $\cond{k'}:= \lambda_1 / \lambda_{k'-1}$ is the condition number.
\end{theorem}

The proof of \Cref{thm:LRA_gap} is identical to \Cref{thm:LRA}, except we use \Cref{thm:kLgood_implies_LRA_gap} in place of \Cref{thm:kLgood_implies_LRA}.

\begin{remark}[Benefits of small block size] \label{rem:small-block-benefits}
  In in the matrix-vector complexity bound of \Cref{thm:LRA_gap}, we observe that the block size $b$ and accuracy $\eps$, but not target rank $k$ or $\ell$ (which is at least $k$), participate in the second term.
  This already suggests that choosing a smaller block size $b$ reduces our overall number of matrix-vector products, supporting the empirical observation in \cref{fig:intro}.
  Our course, we recall from \Cref{sec:matvec-mystery} that the matrix-vector complexity of RBKI is not perfectly reflective of the real-world runtime of RBKI.
\end{remark}

\subsection{Ensuring a singular value gap by perturbing the input} \label{sec:smoothed-analysis}

The careful reader may notice that the RBKI bounds \Cref{thm:LRA,thm:LRA_gap} become vacuous if the minimum relative gap $\mingap{k}$ is small.
Fortunately, the minimum gap can be controlled by adding a random perturbation to the input.
The idea is as follows:
\begin{itemize}
    \item Fix $\gamma > 0$, and generate a diagonal perturbation \mD with entries drawn uniformly from $[-\gamma,\gamma]$.
    \item Run RBKI on the shifted matrix $\mA\mA^\tp + \mD$, obtaining $\mZ = \orth(\cK_q(\mA\mA^\tp + \mD))$.
    \item Output $\hat\mA = \mZ \llbracket \mZ^\tp \mA\rrbracket_k$ , as usual.
\end{itemize}
This perturbation approach was proposed by Meyer, Musco, and Musco \cite[sec.~5]{MMM24} \rev{and analyzed using a well-known result of Minami \cite{Min96}}.
\rev{Concretely, referring to \cite[Lem.~5.1]{MMM24}, we have:}
\begin{importedlemma}[Smoothed analysis]
    Let $\hat{\lambda}_i$ denote the eigenvalues of $\mA\mA^\tp + \mD$, defined above, and assume $\gamma \le \norm{\mA}_2$.
    There exists a universal constant $\const > 0$ such that, with probability $\delta$,
    \begin{equation*}
        \min_{i = 1,\ldots,n-1} \frac{\hat{\lambda}_i - \hat{\lambda}_{i+1}}{\hat{\lambda}_i} \ge \const \cdot \frac{\delta}{n^2} \cdot \frac{\gamma^2}{\norm{\mA}_2^2}.
    \end{equation*}
\end{importedlemma}

To obtain the best possible time complexities, we may choose the perturbation magnitude $\gamma \coloneqq \varepsilon \sigma_{k+1}^2/3n$ as in \cite[Cors.~5.1 \& 5.2]{MMM24}, yielding the matrix-vector complexity $bq = \order( k \log ({n\cond{k'}}/{\delta \varepsilon})/\sqrt{\varepsilon})$, independent of the gap $\mingap{k'}$.
(As in the formal statement of \Cref{thm:LRA}, $k' = b \cdot \lceil k / b \rceil \approx k$.)
In practice, we can choose the perturbations to be on the order of the machine precision $\gamma = \order(\varepsilon_{\rm mach}\norm{\mA}_2)$, which ensures that $\cond{k'} / \mingap{k'} \le \poly(n) / \poly(\varepsilon_{\rm mach})$.
RBKI with random perturbations was used empirically in recent work of Kressner and Shao \cite{KS24}.

\section{Bounding the minimum singular value of a Krylov matrix}
\label{sec:krylov_min_singularvalue}

In this section we prove the full version of \Cref{thm:krylov_sigmamin_bound_intro}, a bound on the minimum singular value of a square random block Krylov matrix.
Our approach involves a simplification and extension of the techniques of Peng and Vempala \cite{PV21} and Nie \cite{Nie22}.

\begin{theorem}[Random block Krylov matrices are not too ill-conditioned]
\label{thm:krylov_sigmamin_bound}
    Let $\mC\in\bbR^{k\times k}$ be a symmetric matrix with distinct eigenvalues $1=\lambda_1 > \cdots > \lambda_k\geq 0$, let $\mH \sim \cN(0,1)^{k\times b}$ be a Gaussian starting block, and set $t= b/k$.
Then there is a universal constant $\const > 0$ such that
\[
\Pr\left[
\sigma_{\rm min}\left( 
\begin{bmatrix}
    \mH & \mC \mH & \cdots & \mC^{t-1}\mH
\end{bmatrix}
\right) \geq \const \cdot \frac{\delta^5}{k^{15}}\left(\frac{\mingap{k}}{6\cond{k}}\right)^{6(t-1)}\right] 
\geq 1-\delta,
\]
where $\mingap{k} := \min_{i=1, \ldots, k-1} (\lambda_i - \lambda_{i+1} )/\lambda_i$ is the minimum (relative) eigenvalue gap and $\cond{k}:= \lambda_1 / \lambda_{k-1}$ is the rank-$k$ condition number.
\end{theorem}
\noindent We have not made a serious attempt to optimize the powers on the terms $\delta$ and $k$ in this expression, opting to present the simplest possible proof with somewhat worse powers.
Indeed, when used to analyze \Cref{alg:block_krylov_LRA}, the lower bound on the minimum singular value appears inside a logarithm, so the key term above is the $\order(t)$ bound in the exponent.

\subsection{Reduction to Vandermonde form}
\Cref{thm:krylov_sigmamin_bound} concerns a square block Krylov matrix 
\begin{equation*}
    \mK_{\rm orig} = \begin{bmatrix}
    \mH & \mC \mH & \cdots & \mC^{t-1}\mH
\end{bmatrix}\in\bbR^{k\times k},
\end{equation*}
where $\mH \sim \cN(0,1)^{k\times b}$ is a Gaussian matrix and \mC is a real symmetric matrix with eigenvalues $1 = \lambda_1 \ge \cdots \ge \lambda_k \ge 0$.
Owing to the rotational invariance of Gaussian matrices and the rotational invariance of singular values, we assume, without loss of generality, that $\mC$ is diagonal.
Then, by permuting the columns of $\mK_{\rm orig}$, we obtain the convenient form
\begin{equation} \label{eq:vandermonde_form}
\mK \defeq \bmat{
        \diag(\vg_1)\mV & \diag(\vg_2)\mV & \cdots & \diag(\vg_b)\mV
    }\in\bbR^{k\times k},
\end{equation}
where $\vg_i\sim \cN(0,1)^k$ are the columns of \mH and \mV is the Vandermonde matrix
\[
\mV \defeq \bmat{
            1 & \lambda_1 & \cdots & \lambda_1^{t-1} \\
            1 & \lambda_2 & \cdots & \lambda_2^{t-1} \\
            \vdots & \vdots & \ddots & \vdots \\
            1 & \lambda_k & \cdots & \lambda_k^{t-1}
        }.
\]
Since $\mK_{\rm orig}$ and $\mK$ are equivalent up to orthogonal transforms, they have the same singular values.
We will prove \Cref{thm:krylov_sigmamin_bound} by manipulating the Vandermonde form \Cref{eq:vandermonde_form}.

\subsection{Block Krylov matrices are almost surely full-rank}

We begin by establishing a simple qualitative result: the Krylov matrix \mK is nonsingular with probability 1.

\begin{proposition}[Block Krylov matrices are nonsingular]\label{thm:block_krylov_fullrank}
Let $\mC \in \bbR^{k\times k}$ be a symmetric matrix with distinct eigenvalues, let $\mH \sim \cN(0,1)^{k\times b}$ be a Gaussian starting block, and set $t = b/k$.
Then $[\begin{matrix} \mH & \mC\mH & \cdots & \mC^{t-1}\mH\end{matrix}]$ is nonsingular with probability 1.
\end{proposition}
We provide a self-contained proof of \Cref{thm:block_krylov_fullrank} via the Schwartz--Zippel lemma in \cref{sec:SZ}, noting that similar results and proofs are common in research on algebraic complexity and linear algebra over finite fields \cite{K95}.

\subsection{\rev{Reduction to subblocks}}
\label{sec:reductions}

To lower bound the minimum singular value of \(\mK\), we employ a method \rev{used} by Peng and Vempala \cite[Sec.~5]{PV21}.
For each $j$, introduce a matrix
\begin{equation*}
\mQ_j\defeq \orth\big(
    \operatorname{range}\big(
        \big[\begin{matrix}
          \diag(\vg_1) \mV
          & \cdots &
          \diag(\vg_{j-1}) \mV
          &
          \diag(\vg_{j+1}) \mV
          & \cdots &
          \diag(\vg_b) \mV
       \end{matrix}\big]
   \big)^\perp \big).
\end{equation*}
Here $\cU^\perp$ indices the orthogonal compliment of a vector space $\cU$, and $\orth(\cV)$ returns any orthonormal basis for a vector space $\cV$.
Since $\mK$ is nonsingular with probability 1 (\Cref{thm:block_krylov_fullrank}), $\mQ_j$ has exactly $t$ columns with probability 1 as well.

\begin{corollary}[Dimensions of $\mQ_j$]\label{fact:Qj_shape}
If $\lambda_1 > \cdots > \lambda_k$, then $\mQ_j\in\bbR^{k\times t}$ with probability 1.
\end{corollary}
Next, we show that bounding the minimum singular value of \mK can be accomplished by ``breaking it into $b$ smaller pieces'' which isolate the dependence on the individual Gaussian \(\vg_i\) vectors in a particularly simple way.
The idea is due to Peng and Vempala \cite[Sec.~5]{PV21}.

\begin{lemma}[Breaking Krylov matrix into pieces] \label{lem:PV_reduction}
It holds that
\begin{equation*}
		\sigma_{\rm min}(\mK) \geq \frac{1}{\sqrt{b}} \min_{1\le j \le b} \sigma_{\rm min}\left(\mQ_j^\tp \diag(\vg_j) \mV\right).
\end{equation*}    
\end{lemma}

\begin{proof}
Let $\valpha = (\valpha_1,\ldots,\valpha_b) \in \bbR^k$ be a unit vector, chunked into pieces $\valpha_j$ of size $t$.
Then, since $\mQ_j$ has columns orthogonal to the columns of $\diag(\vec{g}_i)\mV$ for $i \ne j$,
\[
\norm{\mQ_j^\tp \mK\valpha}_2 
    = \norm[\bigg]{\mQ_j^\tp \sum_{i=1}^b \diag(\vg_i) \mV\valpha_i}_2 
    = \norm{\mQ_j^\tp \diag(\vg_j) \mV \valpha_j}_2 .
\]
Therefore, since $\mQ_j$ has orthonormal columns, we have
\begin{align*}
	\norm{\mK\valpha}_2 
    \ge \norm{\mQ_j^\tp \mK\valpha}_2 
    = \norm{\mQ_j^\tp \diag(\vg_j) \mV \valpha_j}_2
    \ge \sigma_{\rm min}\left(\mQ_j^\tp \diag(\vg_j) \mV\right) \norm{\valpha_j}_2.
\end{align*}
Next, since $\valpha$ is a unit vector, there exists some $j_\star$ for which $\norm{\valpha_{j_\star}} \ge 1/\sqrt{b}$. 
Thus, instantiating this bound for $j_\star$, we obtain
\begin{equation*}
    \norm{\mK\valpha}
    \ge
        \frac{1}{\sqrt{b}}
        \sigma_{\rm min}\left(
            \mQ_{j_\star}^\tp \diag(\vg_{j_\star}) \mV
        \right)
    \ge
        \frac{1}{\sqrt{b}}
        \min_{1 \le j \le b}
        \sigma_{\rm min}\left(
            \mQ_j^\tp \diag(\vg_j) \mV
        \right) .
\end{equation*}
Minimizing over unit vectors $\valpha$ on the left gives the desired result.
\end{proof}

\subsection{Anti-concentration and non-sparsification}\label{sec:anti-non-sparse}
In our goal to lower bound the minimum singular value of \mK, \cref{lem:PV_reduction} shows that it suffices the bound the minimum singular values \(\sigma_{\min}(\mQ_j^\tp \diag(\vg_j) \mV)\) for all \(j=1,\ldots,b\).
To bound the minimum singular value of each $\mQ_j^\tp \diag(\vg_j) \mV$, we use anti-concentration theorem of Nie \cite{Nie22}.
Given a Gaussian matrix \mZ, Nie's theorem states that if \(\vx^\tp\mZ\vy\) bounded away from zero with good probability for every pair of unit vectors $\vx$ and $\vy$, then the minimum singular value of \mZ is not too small with good probability as well.

\begin{importedtheorem}[Matrix anti-concentration, \protect{\cite[Thm.~1.6]{Nie22}}]
\label{impthm:nie-small-ball}
Let \( \mZ\in\bbR^{t \times t} \) be real random matrix with jointly Gaussian entries.
Suppose that for some \( \eta \geq 0 \), it holds that
\[
\Pr\left[ |\vx^\tp\mZ\vy| > \eta \right] \geq \frac{1}{2} \quad \text{for all unit vectors } \vx,\vy \in \bbR^t.
\]
Further suppose that \( \Pr[\norm{\mZ}_2 > M] \leq \tfrac{1}{8} \).
Then there exists a universal constant $\const > 0$ such that
\[
\Pr \left[ \sigma_{\rm min}(\mZ) \leq \const \cdot \frac{\delta^2 \eta^2} { t^3 M } \right] \leq  \delta.
\]
\end{importedtheorem}

We will use this theorem to bound the minimum singular value of \(\mQ_j^\tp \diag(\vg_j) \mV\).
Fix any unit vectors $\vx$ and $\vy$.
Conditional on $\mQ_j$, the bilinear form $\vx^\tp \mQ_j^\tp \diag(\vg_j) \mV \vy$ has a Gaussian distribution:
\begin{equation} \label{eq:bilinear-form-gaussian}
    \vx^\tp \mQ_j^\tp \diag(\vg_j) \mV \vy \mid \mQ_j
    \sim \cN\left(0, ~ \sum\nolimits_{i=1}^k \left[\mQ_j\vx\right]_i^2 \left[\mV\vy\right]_i^2  \right).
\end{equation}

The small-ball condition in Nie's result asks us to show that this Gaussian is bounded away from zero with high probability.
As \cref{lem:gaussian_anti_concentration} shows, a centered Gaussian is unlikely to be much smaller than its standard deviation.
As such, we just need to show that the variance in \cref{eq:bilinear-form-gaussian} is not too small.
To that end, we prove that for all vectors \vx and \vy, \(\mQ_j\vx\in\bbR^k\) has at least \(t+1\) entries bounded away from zero and \(\mV\vy\in\bbR^k\) has at least \(k-(t-1)\) entries bounded away from zero.
Since \((t+1)+(k-(t-1)) > k\), we know that there exists at least one index \(i_\star\) such that both \([\mQ_j\vx]_{i_\star}\) and \([\mV\vy]_{i_\star}\) are bounded away from zero.
Then the variance in \cref{eq:bilinear-form-gaussian} is at least \([\mQ_j\vx]_{i_\star}^2[\mV\vy]_{i_\star}^2\), which is also bounded away from zero.
\Cref{thm:krylov_sigmamin_bound} will then follow from Nie's \cref{impthm:nie-small-ball}.

We now execute this approach.
First, we state ``non-sparsification'' guarantees for \(\mV\) and \(\mQ_j\):

\begin{lemma}[Vandermonde non-sparsification]
\label{lem:V-not-sparse}
Let \(\vy\in\bbR^t\) be a vector.
Then \(\mV\vy\) has at least \(k-(t-1)\) entries of magnitude at least \(\norm{\vy}_1/t \cdot ( \mingap{k} / {6\cond{k}})^{2(t-1)}\).
\end{lemma}

\begin{lemma}[Subspace non-sparisification]
\label{lem:Q-not-sparse}
Let \(\vx\in\bbR^t\) be a unit vector.
Then, with probability at least \(1-\delta \), \(\mQ_j\vx\) has at least \(t+1\) entries of magnitude at least \(({\delta\log^{-1/2}(2/\delta)}/{14 k^3})  (\mingap{k}/2\cond{k})^{t-1}\).
\end{lemma}

\Cref{lem:V-not-sparse} is similar to \cite[Lem.~5.1]{PV21}, and \cref{lem:Q-not-sparse} is an improvement on previous results (see \cref{sec:Q-not-sparse} for discussion).
We will prove these lemmas in \cref{sec:V-not-sparse,sec:Q-not-sparse}, respectively.
For now, we use these results to prove \cref{thm:krylov_sigmamin_bound}:

\begin{proof}[Proof of \Cref{thm:krylov_sigmamin_bound}]
Our strategy is as follows.
First, we will bound $\sigma_{\rm min}(\mQ_j^\tp\diag(\vg_j) \mV)$ separately for each $j$ using Nie's anti-concentration theorem (\Cref{impthm:nie-small-ball}).
Then, we will union bound over all $j$ and invoke the Peng--Vempala bound (\Cref{lem:PV_reduction}) for $\sigma_{\rm min}(\mK)$.

Fix a value of $j$, define $\mZ_j \defeq \mQ_j^\tp\diag(\vg_j) \mV$, and let $\Pr_j$ indicate probabilities over only the randomness in \(\vg_j\) (i.e., conditioned on the values of $\{\vg_i : i \ne j\}$).
To invoke Nie's anti-concentration theorem (\Cref{impthm:nie-small-ball}), we need to identify appropriate parameters $M$ and $\eta$.

\vspace{0.4em}
\noindent \textbf{\textit{The spectral norm bound $M$.}}
The easier input to Nie's theorem is the parameter $M$, defined as a number for which $\Pr_j[\norm{\mZ_j}_2 > M] \le 1/8$.
By submultiplicativity, $\norm{\mZ_j}_2$ is bounded as
\begin{equation*}
    \norm{\mZ_j}_2 \le \norm{\mQ_j}_2 \norm{\diag(\vg_j)}_2 \norm{\mV}_2 \le \norm{\vg_j}_2\norm{\mV}_\fro \le k\norm{\vg_j}_2.
\end{equation*}
In the last inequality, we used the fact that $\mV$'s entries are all in $[0,1]$.
Therefore, by the norm tail bound \Cref{lem:chisq_concentration}, we obtain
\begin{equation} \label{eq:M-def}
    \Pr_j\left[\norm{\mZ_j}_2 > M\right] = \Pr_j\left[\norm{\vg_j}_2^2 > 25k\right] < \exp(-5) < \frac{1}{8} \quad \text{for } M = 5k^{3/2}.
\end{equation}

\vspace{0.4em}
\noindent \textbf{\textit{The small-ball parameter $\eta$.}}
Fix unit vectors $\vx$ and $\vy$, and a failure probability $\delta' > 0$ to be set later.
Our goal is to find a parameter $\eta>0$ for which $\Pr_j[|\vx^\tp \mZ_j\vy|] = \Pr_j[|(\mQ_j\vx)^\tp \diag(\vg_j)(\mV \vy)| > \eta] \ge 1/2$.
To do so, we invoke our non-sparsification results:
\begin{itemize}
    \item Since $\norm{\vy}_2 = 1$, it must hold that $\norm{\vy}_1 \ge 1$.
    Thus by \Cref{lem:V-not-sparse}, $\mV\vy$ has at least $k-t-1$ entries of magnitude at least \(\norm{\vy}_1/t \cdot (\mingap{k} / {6\cond{k}})^{2(t-1)} \ge k^{-1} (\mingap{k} / {6\cond{k}})^{2(t-1)}\).
    \item By \Cref{lem:Q-not-sparse}, with probability at least $1-\delta'$, $\mQ_j\vx$ has at least $t+1$ entries of magnitude at least $(\delta'\log^{-1/2}(2/\delta')/{14 k^3}) (\mingap{k}/2\cond{k})^{t-1}$.
    We condition on this ``good event'' going forward.
\end{itemize}
The total number of ``large'' entries in $\mQ_j\vx$ and $\mV\vy$ is $(t+1) + (k-t+1) = k+2$.
Therefore, by the pigeonhole principle, there must exist at least one index \(1\le i_\star \le k\) such that
\begin{equation*}
    |[\mV\vy]_{i_\star}| \ge \frac{1}{k} \left(
    \frac{\mingap{k}}{{6\cond{k}}}\right)^{2(t-1)} \quad \text{and} \quad |[\mQ_j\vx]_{i_\star}| \ge \frac{\delta'\log^{-1/2}(2/\delta')}{14 k^3}  \left(
    \frac{\mingap{k}}{2\cond{k}}\right)^{t-1}.
\end{equation*}
Thus, by \Cref{eq:bilinear-form-gaussian}, the conditional distribution of $\vx^\tp \mZ_j \vy$ with respect to $\mQ_j$ is a Gaussian random variable with mean and zero and standard deviation
\begin{equation} \label{eq:eta-def}
    \sqrt{\Var(\vx^\tp \mZ_j \vy \mid \mQ_j)} \ge \frac{\delta'\log^{-1/2}(2/\delta')}{14 k^4}  \left(\frac{\mingap{k}}{6\cond{k}}\right)^{3(t-1)}.
\end{equation}
Ergo, by Gaussian anti-concentration (\Cref{lem:gaussian_anti_concentration}),
\begin{equation*}
    \Pr_j\left[|\vx^\tp \mZ_j \vy| > \eta \right] \ge \frac{1}{2} \quad \text{for } \eta \coloneqq \frac{\delta'\log^{-1/2}(2/\delta')}{28 k^4}  \left(\frac{\mingap{k}}{6\cond{k}}\right)^{3(t-1)}.
\end{equation*}

\vspace{0.4em}
\noindent \textbf{\textit{Finishing up: Invoking Nie's theorem and union bounding.}}
By Nie's anti-concentration theorem (\Cref{impthm:nie-small-ball}), it holds that
\begin{equation*}
    \Pr_j\left[\sigma_{\rm min}(\mZ_j) \le \const \frac{(\delta')^2\eta^2}{t^3M}\right] \le \delta',
\end{equation*}
where $M$ and $\eta$ are defined in \Cref{eq:M-def,eq:eta-def}.
To complete the proof of the theorem, we need to union bound over this bound holding and the conclusion of \Cref{lem:Q-not-sparse} holding for each $j = 1,\ldots,b$.
Thus, setting $\delta' \defeq {\delta}/{2b}$ ensures all necessary ``good'' events hold.
Therefore, by \Cref{lem:PV_reduction},
\begin{equation*}
    \sigma_{\rm min}(\mK) \ge \frac{1}{\sqrt{b}} \min_{1\le j \le b} \sigma_{\rm min}(\mZ_j) > \const \frac{(\delta')^2\eta^2}{\sqrt{b}t^3M} \iffalse = \const \frac{\delta^4\log^{-1}(b/\delta)}{b^{4.5}t^3k^{8.5}} \left(\mingap{k}/6\cond{k}\right)^{6(t-1)}\fi
\end{equation*}
holds with probability at least $1-\delta$.
Substituting \cref{eq:M-def,eq:eta-def}, simplifying, and using the crude bounds $b,t \le bt = k$ yields the stated result.
\end{proof}

\subsection{Vandermonde non-sparsification}
\label{sec:V-not-sparse}

We now prove \Cref{lem:V-not-sparse}, which shows that the Vandermonde matrix \mV is a non-sparsifier.
Our proof relies on the following result:

\begin{importedtheorem}[Vandermonde matrices are not too ill-conditioned \protect{\cite[Thm.~1]{Gau62}}]
\label{impthm:gautschi}
Let \(\mW\in\bbR^{t \times t}\) be a square Vandermonde matrix associated with nodes \(1\ge \mu_1 \geq \ldots \geq \mu_t \geq 0\).
Let \(\eta \defeq \min_{i=1,\ldots,t-1} \mu_i - \mu_{i+1}\) be the minimum additive gap between nodes.
\[
    \norm{\mW^{-1}}_1
    \defeq \max_{\vbeta\neq \vec{0}} \frac{\norm{\mW^{-1}\vbeta}_1}{\norm{\vbeta}_1}
    \leq \max_{i=1,\ldots, t} \, \prod_{\substack{j=1 \\ j\neq i}}^t ~ \frac{1+\mu_j}{|\mu_i-\mu_j|}
    \leq \left(\frac{2}{\eta}\right)^{t-1}.
\]
\end{importedtheorem}
We can then prove the associated non-sparsification lemma:

\begin{proof}[Proof of \Cref{lem:V-not-sparse}]
By homogeneity, we may assume $\norm{\vy}_1 = 1$.
Let \(p_{\vy}(x) = \sum_{i=1}^t [\vy]_i x^{i-1}\) be the polynomial with coefficients described by \vy.
In particular, notice that \([\mV\vy]_i = p_{\vy}(\lambda_i)\).
Let \(\mW\in\bbR^{t \times t}\) be the Vandermonde matrix formed by nodes \(\lambda_1, \lambda_{2}, \cdots, \lambda_{t}\). 
Observe that $\mW$ consists of the first $t$ rows of $\mV$, and define $\eta \coloneqq \max_{i=1,\ldots,k-1} \lambda_i - \lambda_{i+1}$, which is at least $\mingap{k} / \cond{k}$ by \Cref{lem:additive_gap_to_Delta}.

By \Cref{impthm:gautschi}, for any vector $\vbeta$ it holds that \(\norm{\vbeta}_1 \geq ({\eta}/{2})^{t-1}\norm{\mW^{-1}\vbeta}_1\).
Then, since \mW is a subset of the rows of \( \mV\), we have that
\[
    \norm{\mV\vy}_1
    \geq \norm{\mW\vy}_1
    \geq \left(\frac{\eta}{2}\right)^{t-1} \norm{\mW^{-1}\mW\vy}_1
    = \left(\frac{\eta}{2}\right)^{t-1}.
\]
That is, we are guaranteed that the absolute sum of the entries in \(\mV\vy\) is at least \(({\eta}/{2})^{t-1}\).
Consequently, there exists $i_\star\in\{1,\ldots, t\}$ such that 
\[
|p_{\vy}(\lambda_{i_\star})| \geq \frac{1}{t}\left(\frac{\eta}{2}\right)^{t-1}.
\]
Next, we expand \( p_{\vy}(x) \) in terms of its \(T \leq t-1\) zeros \(z_1,\cdots,z_T\):
\[
    p_{\vy}(x)
    = p_{\vy}(\lambda_{i_\star}) \prod_{j=1}^T \frac{x-z_j}{\lambda_{i_\star}-z_j}.
\]
Around each zero \(z_j\), introduce a disk \(D_{j}\) of radius \( r \defeq \eta/3 \).
Since the diameter of $D_j$ is $2\eta/3 < \eta$ and all eigenvalues are at least a distance $\eta$ apart, at most 1 eigenvalue lives in $D_j$.
Since there are at most \(t-1\) disks \(D_j\), there are at least \(k-(t-1)\) eigenvalues \(\lambda_\ell\) not in any disk.
For such \(\lambda_\ell\), 
\[
    \frac{|\lambda_\ell - z_j|}{|\lambda_{i_\star} - z_j|}
    \geq\frac{|\lambda_\ell - z_j|}{|\lambda_{i_\star} - \lambda_\ell| + |\lambda_\ell - z_j|}
    \geq \frac{|\lambda_\ell - z_j|}{1 + |\lambda_\ell - z_j|}
    \geq \frac{r}{1 + r}
    \geq \frac{r}{2}
    = \frac{\eta}{6},
\]
where we have used the triangle inequality, the fact that $|\lambda_{i_\star} - \lambda_\ell| \leq 1$ since all eigenvalues are $\leq 1$, the fact $x \mapsto \frac{x}{1+x}$ is an increasing function, and the fact that $r<1$.
Therefore,
\[
    |p_{\vy}(\lambda_\ell)|
    = |p_{\vy}(\lambda_{i_\star})| \prod_{j=1}^T \left|\frac{\lambda_\ell - z_j}{\lambda_{i_\star} - z_j}\right|
    \geq \frac{1}{t} \left(\frac{\eta}{2}\right)^{t-1} \prod_{j=1}^T \left(\frac{\eta}{6}\right)
    \geq \frac{1}{t}\left(\frac{\eta}{6}\right)^{2(t-1)}
    \geq 
    \frac{1}{t}\left(\frac{\mingap{k}}{6\cond{k}}\right)^{2(t-1)}.
\]
This completes the proof.
\end{proof}

\subsection{Subspace non-sparsification}
\label{sec:Q-not-sparse}

In this section, we will prove \Cref{lem:Q-not-sparse}. 
It will be useful to make the following definition.
\begin{definition}[Non-sparsification certificate]
    A matrix \(\mGamma \in \bbR^{k \times m}\) is a \emph{\((s,\eta,M)\)-non-sparsification certificate} for \(\mQ\) if it satisfies the following three properties:\begin{enumerate}
        \item
            \label[property]{prop:ortho}
            \textbf{Orthogonality:~}
            \(\mGamma^\tp\mQ=\mat0\).
        \item
            \label[property]{prop:norm}
            \textbf{Norm:~}
            \(\norm{\mGamma}_2 \leq M\).
        \item
            \label[property]{prop:rip}
            \textbf{Injective RIP:~}
            \(\norm{\mGamma^\tp\vu}_2 \geq \eta\norm\vu_2\) for all \(s\)-sparse vectors \(\vu\).
    \end{enumerate}
\end{definition}

We now show that any non-sparsification certificate implies that \(\mQ\vx\) has many large entries.

\begin{lemma}[Non-sparsification certificates certify non-sparsification]
    \label{lem:abstract-q-nonsparse}
    Suppose there exists a matrix \mGamma that is a \((s,\eta,M)\)-nonsparsification certificate for \(\mQ\).
    Let \(\vx\in\bbR^t\) be a unit vector.
    Then \(\mQ\vx\) has at least \(s+1\) entries of magnitude at least \(\eta/(3M\sqrt{k})\).
\end{lemma}
\begin{proof}
    The result is trivial if $\eta = 0$, so assume $\eta > 0$.
    Let \(\vu \defeq \mQ\vx \in \bbR^k\).
    Assume, for sake of contradiction, that \(\vu\) has at most \(s\) entries with magnitude greater than or equal to \(\gamma \defeq \eta/(3M\sqrt{k}) \).
    Decompose \(\vu = \vu_1 + \vu_2\) where \(\vu_1\) restricts \vu to its \(s\) largest magnitude entries and is zero elsewhere.
    By construction, every entry of \(\vu_2\) is smaller than \(\gamma\) in magnitude, so
    \[
        \norm{\vu_2}_2 \leq \sqrt k \gamma \leq \frac{\eta}{3M}.
    \]
    Since \( \mQ\) has orthonormal columns, \(\norm{\vu}_2 = \norm{\vx}_2 = 1\).
    Then, by the triangle inequality
    \[
    \norm{\vu_1}_2 \geq \norm{\vu}_2 - \norm{\vu_2}_2 \geq 1 - \norm{\vu_2}_2 \geq 1 - \frac{\eta}{3M}.
    \]
    Then, by the injective RIP and norm properties, we know that
    \begin{equation*}
        \norm{\mGamma^\tp\vu_1}_2 \geq \eta \norm{\vu_1}_2 \geq \eta\left(1-\frac{\eta}{3M}\right) 
        \quad\text{and}\quad
        \norm{\mGamma^\tp\vu_2}_2
        \leq \norm{\mGamma}_2\norm{\vu_2}_2 \leq M \norm{\vu_2}_2 \leq \frac\eta3.
    \end{equation*}
    Therefore, by reverse triangle inequality, 
    \[
        \norm{\mGamma^\tp\mQ\vx} = \norm{\mGamma^\tp\vu}
        \geq \norm{\mGamma^\tp\vu_1} - \norm{\mGamma^\tp\vu_2}
        \geq \eta\left(1-\frac{\eta}{3M}\right) - \frac\eta3
        = \eta \left(\frac{2}{3} - \frac{\eta}{3M}\right)
        > 0.
    \]
    In the last line, we use the fact that $\eta \le \norm{\mGamma}_2 \le M$.
    We have derived a contradiction with the orthogonality property, which asserts $\mGamma^\tp\mQ = \vec{0}$.
    We conclude that \(\vu=\mQ\vx\) indeed has more than \(s\) entries of magnitude at least \(\gamma\).
\end{proof}

To prove \Cref{lem:Q-not-sparse}, it suffices to find an appropriate non-sparsification certificate for $\mQ_j$.
Any submatrix of
\begin{equation*}
    \begin{bmatrix}
          \diag(\vg_1) \mV
          & \cdots &
          \diag(\vg_{j-1}) \mV
          &
          \diag(\vg_{j+1}) \mV
          & \cdots &
          \diag(\vg_b) \mV
       \end{bmatrix},
\end{equation*}
yields \emph{a} non-sparsification certificate for $\mQ_j$, since $\mQ_j$ is orthogonal to the range of this matrix by construction.
But proving \Cref{lem:Q-not-sparse} requires identifying the ``right'' non-sparsification certificate.
Though not explicitly written as such, the previous works of Peng and Vempala \cite{PV21} and Nie \cite{Nie22} can be interpreted as choosing $[\begin{matrix}
    \vg_1 & \cdots & \vg_{j-1} & \vg_{j+1} & \cdots & \vg_b
\end{matrix}]$ as the non-sparsification certificate.
For this matrix, the injective RIP property holds only with sparsity \(s= \order(b / \log k)\).
Since this only yields useful results on the complexity of RBKI when \(s \geq t\), this choice of non-sparsification certificate is informative only when \(b = \Omega(\sqrt{k \log k})\).
This square-root dependence of the block-size on the subspace dimension $k$ is visible in Nie's \cite[Thm.~1.10]{Nie22}.
To prove \Cref{lem:Q-not-sparse} with no restriction on block size, we make the other natural choice for the non-sparsification certificate, the \emph{single-vector Krylov matrix} $\diag(\vg_i) \mV$ associated with any index $i\ne j$.

\begin{lemma}[Single-vector Krylov matrix as non-sparsification certificate]\label{thm:V-non-sparse-cert}
Fix $i \ne j$. 
With probability at least $1-\delta$,
\[
    \mGamma \defeq \diag(\vg_i) \mV
\]
is a $(s,\eta,M)$-non-sparsification certificate for \(\mQ_j\) with 
\[ 
s=t
, \quad 
\eta= \frac{\delta}{2k^{3/2}}  \left(\frac{\mingap{k}}{2\cond{k}}\right)^{t-1}
,\quad  \text{and}\quad
M = \sqrt{5} k^{3/2} \sqrt{\log(2/\delta)}.
\]
\end{lemma}

\noindent Combining this result with \Cref{lem:abstract-q-nonsparse} immediately proves \Cref{lem:Q-not-sparse}.
We now prove \Cref{thm:V-non-sparse-cert}:

\begin{proof}
We must verify all three properties of a non-sparsification certificate.
The orthogonality property holds by construction, and we can obtain the norm property by applying the crude bound $\norm{\mGamma}_2 \leq \norm{\vg_j}_2\norm{\mV}_\fro \leq k\norm{\vg_j}_2$ and using \Cref{lem:chisq_concentration}:
\[
\Pr\left[\norm{\mGamma}_2 \geq \sqrt{5} k^{3/2} \sqrt{\log(2/\delta)} \right]
\leq \Pr\left[ \norm{\vg_j}^2 \geq 5k \log(2/\delta) \right]
\leq \delta/2.
\]

We now prove the injective RIP property.
For notational convenience, write $\vg = \vg_i$.
We must show that \(\norm{\mGamma^\tp\vu}_2\geq\eta\) holds for all \(t\)-sparse unit vectors \vu.
Fix a $t$-sparse unit vector $\vu$ and define $\vz \defeq \vu \odot \vg$, where ``\(\odot\)'' denotes the Hadamard, or entrywise, product between vectors.
Observe that
\[
    \norm{\mGamma^\tp\vu}_2
    = \norm{\mV^\tp\diag(\vg)\vu}_2
    = \norm{\mV^\tp(\vg \odot \vu)}
    = \norm{\mV^\tp \vz}_2.
\]
By Gaussian anti-concentration (\Cref{lem:gaussian_anti_concentration}) and a union bound,
\[
\Pr\left[\min_{1\le\ell \le k}|[\vg]_\ell| < \frac{\delta}{2k} \right]
\leq 
k \Pr\left[ |[\vg]_\ell| < \frac{\delta}{2k}\right]
\leq \frac{\delta}{2}.
\]
For the remainder of this proof, condition on the event that this ``bad event'' does not happen.
Conditioned as such, we have,
\[
\norm{\vz}_2 
= \norm{\vu \odot \vg}_2 
\geq \frac{\delta}{2k}\norm{\vu}_2
= \frac{\delta}{2k}.
\]
Let $i_1,\ldots,i_t$ denote a set of indices on which $\vz$ is supported, and let $\vy \in \bbR^t$ denote the restriction of $\vz \in \bbR^k$ to these entries.
Similarly, let $\mW$ denote the row submatrix of $\mV$ indexed by $i_1,\ldots,i_t$.
By construction,\mW is square and Vandermonde, $\mV^\tp \vx = \mW^\tp \vy$, and $\norm{\vz}_2 = \norm{\vy}_2 \ge \delta/2k$.

By \Cref{impthm:gautschi} (and \Cref{lem:additive_gap_to_Delta}), the $\ell_1$ operator-norm of $\mW^{-1}$ is bounded:
\[
    \norm{\mW^{-1}}_\infty \leq \left(\frac{2}{\min_{i=1,\ldots,k-1} \lambda_i - \lambda_{i+1}}\right)^{t-1}
    \leq \left(\frac{2\cond{k}}{\mingap{k}}\right)^{t-1}.
\]
The $\ell_2$ and $\ell_1$ operator norms for linear maps $\bbR^{k}\to\bbR^k$ are comparable, $\norm{\,\cdot\,}_2 \le \sqrt{t}\norm{\,\cdot\,}_1$.
Thus,
\begin{equation*}
    \norm{\mGamma^\tp \vu}_2 = \norm{\mV^\tp \vz}_2 = \norm{\mW^\tp\vy}_2 \ge \norm{\mW^{-1}}_2^{-1}\norm{\vy}_2 \ge \frac{1}{\sqrt{k}} \left(\frac{\mingap{k}}{2\cond{k}}\right)^{t-1}\cdot \frac{\delta}{2k}.
\end{equation*}
This is the advertised result.
\end{proof}

\section{Conclusion and open problems} \label{sec:conclusion}

In this paper, we showed that randomized block Krylov iteration solve low-rank approximation problems with any block size \( 1\leq b\leq k\) using just \(\orderish(k/\sqrt{\eps})\) matrix-vector products, improving on prior bounds which required \(\orderish(k^2/\sqrt\eps)\) in the worst case.
Our results provide theoretical support for the widespread use of intermediate block sizes $1 < b < k$ in applied computation.

There are several possibilities for future work.
First, the dependence of our results on the singular value gaps is more pessimistic than those of \cite{MMM24}, and it is natural to seek a ``best of both worlds'' result that has our dependence on the block size $b$ and \cite{MMM24}'s dependence on the singular value gaps.
We conjecture that our main results \cref{thm:LRA,thm:LRA_gap,thm:krylov_sigmamin_bound} hold with the gap $\mingap{k} = \min_{i = 1,\ldots,k-1} (\sigma_i^2 - \sigma_{i+1}^2)/\sigma_i^2$ replaced with the $b$-th order gap $\smash{\Delta_k^{(b)}} \defeq \min_{i = 1,\ldots,k-b} (\sigma_i^2 - \sigma_{i+b}^2)/\sigma_i^2$ and with the condition number $\cond{k}$ removed entirely.
We regard this as an interesting mathematical question but one of limited practical importance, as the singular value gap $\mingap{k}$ and condition number $\cond{k}$ appear in a logarithm and are easily controlled by randomly perturbing the input (\cref{sec:smoothed-analysis}).
Second, we believe there are many interesting questions regarding the behavior of block Krylov subspace methods in finite-precision arithmetic.
Given the empirical success of RBKI in finite precision, we are optimistic that our results could transfer to that setting with enough care.
Third, we are interested in whether our results and techniques can be extended to the setting of sparse Gaussian starting blocks as in the works of Peng, Vempala, and Nie \cite{Nie22,PV21}.

%% file: acknowledgements.tex
Christopher Musco was partially supported by NSF Awards \#2045590 and \#2427362.
Raphael Meyer and Ethan Epperly received support, under aegis of Joel Tropp, from Office of Naval Research BRC award N00014-24-1-2223, a Caltech Center for Sensing to Intelligence grant, and the Caltech Carver Mead New Adventures Fund.
Ethan Epperly was also supported by the DOE CSGF under award DE-SC0021110.
We thank Cameron Musco, Jerry Li, Joel Tropp, and Robert Webber for helpful discussions about this project. 

%% file: appendix.tex
\section{Proof of \Cref{lem:kpLgood_klgood}}
\label{sec:proofs}

Suppose $\mB$ is $(k',L)$-good for $\mM$. 
Then there exists a matrix $\mQ'$ whose columns live in $\range(\mB)$ such that $\norm{(\mU_{k'}^\tp \mQ')^{-1}}_2 \leq L$.
This implies $\mU_{k'}^\tp \mQ'$ is invertible. 
Let $\mC\in\bbR^{k'\times k}$ be the first $k$ columns of $(\mU_{k'}^{\tp}\mQ')^{-1}$ so that 
\[
(\mU_{k'}^\tp \mQ') \mC = \bmat{\mI_k \\ \mat{0}}.
\]
Let \mV be an orthonormal basis for $\range(\mC)$ and define $\mQ = \mQ' \mV$. 
Note also that $\mU_k^\tp = [\begin{matrix}\mI_k & \mat{0}\end{matrix}] \mU_{k'}^\tp$.
Then, since $\mC = \mV (\mV^\tp\mC)$,
\[
\mU_k^\tp \mQ
= [\begin{matrix}\mI_k & \mat{0}\end{matrix}] \mU_{k'}^\tp \mQ' \mV
= [\begin{matrix}\mI_k & \mat{0}\end{matrix}] \mU_{k'}^\tp \mQ' \mC (\mV^\tp \mC)^{-1}
=(\mV^\tp \mC)^{-1}.
\]
Hence, 
\begin{align}
    \norm{(\mU_k^\tp\mQ)^{-1}}_2^2
= \norm{\mV^\tp\mC}_2^2
\leq \norm{\mC}_2^2
= \norm*{ (\mU_{k'}^\tp\mQ')^{-1}\bmat{\mI_k \\ \mat{0}}}_2^2
\leq \norm*{ (\mU_{k'}^\tp\mQ')^{-1}}_2^2
\leq L.
% \tag*{\qedsymbol}
\end{align}

\section{Nonsingularity of random block Krylov subspaces}
\label{sec:SZ}

In this section we prove \Cref{thm:block_krylov_fullrank}.
We begin by recalling a simplified version of the Schwartz--Zippel lemma \cite{Zip79,Sch80}. 

\begin{importedtheorem}[Schwartz--Zippel]\label{thm:schwartz_zippel}
    Suppose $p: \bbR^q \to \bbR$ is a nonzero polynomial of finite total degree and $\vec{x}$ is a random Gaussian vector. 
    Then, with probability one, $p(\vec{x}) \neq 0$.
\end{importedtheorem}

\begin{proof}[Proof of \Cref{thm:block_krylov_fullrank}]
    We begin by relating the block Krylov subspace to a multivariate polynomial of the entries of the starting block as follows.
    Define the matrix $\mK_t(\mH)$ by 
    \begin{equation}
        \mK_t(\mH) = 
        \begin{bmatrix}
            \mH & \mC\mH & \cdots & \mC^{t-1}\mH
        \end{bmatrix}\in\bbR^{k\times bt}.
    \end{equation}
    Define the polynomial $p:\bbR^{d \times t} \to \bbR$ by
    \begin{equation*}
        p(\mH) = \det(\mK_t(\mH)^\tp \mK_t(\mH)).
    \end{equation*}
    Since $k = bt$, the matrix $\mK_t(\mH)$ is singular is if and only if $p(\mH)\neq 0$.
    Therefore, if we can show that $p:\bbR^{d \times k} \to \bbR$ is not the zero polynomial, then the result follows immediately from the Schwartz--Zippel Lemma (\Cref{thm:schwartz_zippel}).
    
    To show $p:\bbR^{d \times k} \to \bbR$ is not the zero polynomial, we must exhibit a matrix $\hat{\mH}$ such that $p(\hat{\mH}) \neq 0$.
    Work in an orthonormal basis of eigenvectors of $\mC$ so that 
    \begin{equation*}
        \mC = \diag(\mC_1,\ldots,\mC_b)
        \quad\text{where}\quad
        \mC_i \defeq \diag(\lambda_{(i-1)t+1},
        \lambda_{(i-1)t+2}, \ldots, \lambda_{(i-1)t + t}).
    \end{equation*}
    We have used $\diag$ also to indicate a \emph{block} diagonal matrix with the specified entries.
    Define 
    \begin{equation*}
        \hat{\mH} \defeq
        \diag(\vec{1},\ldots,\vec{1}).
    \end{equation*}
    Observe then that, for some permutation matrix $\mat{\Pi}$ 
    \begin{equation*}
        \mK_t(\hat{\mH})
        = \diag(\hat{\mK}_1,\hat{\mK}_2,\ldots,\hat{\mK}_b)
        \qquad
        \text{with }\hat{\mK}_i \defeq \begin{bmatrix}
            \vec{1} & \mC_i\vec{1} & \cdots & \mC_i^{t-1}\vec{1}
        \end{bmatrix}.
    \end{equation*}
    Suppose $\hat{\mK}_t\vc = \vec{0}$, for $\bv{c} = (c_0, c_1, \ldots, c_{t-1})$.
    This means that $c_0 x + c_1 \lambda + \cdots + c_{t-1} \lambda^{t-1} = 0$ for each of the  $t$ eigenvalues $\lambda$ of $\mC_i$.
    But a nonzero polynomial of degree $<t$ cannot have $t$ roots, so $\vc = \vec{0}$.
    This means $\operatorname{rank}(\hat{\mK}_i) 
    = t$.
    Owing to the structure of $\mK_t(\hat{\mH})$, 
    \[
    \operatorname{rank}(\mK_t(\hat{\mH})) =
        \operatorname{rank}(\hat{\mK}_1) + \cdots + \operatorname{rank}(\hat{\mK}_b)
        = tb
        = k.
    \]
    This completes the proof.
\end{proof}

% \clearpage

\section{Additional numerical experiments} \label{sec:more-figs}

\begin{figure}[t]
    \centering
    \includegraphics[scale=0.8]{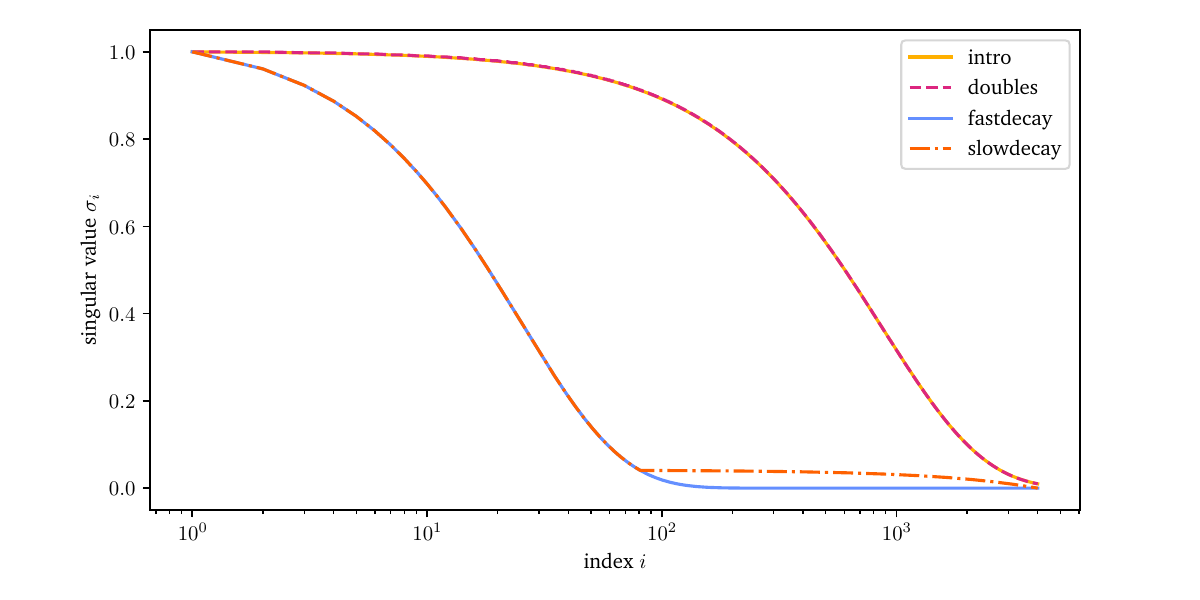}
    \caption{Spectra for the \ref{eq:intro}, \ref{eq:doubles}, \ref{eq:fastdecay}, and \ref{eq:slowdecay} problems.}
    \label{fig:spectra}
\end{figure}

In this section, we provide additional numerical experiments and further details on the experiment displayed in \cref{fig:intro}.
We test square matrices $\mat{A}$ of dimension $n = 4000$ with four singular value profiles, defined for $i = 1,\ldots,n$:
% \begin{equation}
% \label{eqn:spectrums}
% % \sigma_i :=
% % \begin{cases}
% %   0.01^{\tfrac{i-1}{\,n-1}} & \texttt{intro}, \\[6pt]
% %   0.01^{\tfrac{\lceil i/2 \rceil - 1}{\,n/2 - 1}} & \texttt{doubles}, \\[6pt]
% %   \exp(-i/25) & \texttt{fastdecay}, \\[6pt]
% %   \max\!\left\{ \exp(-i/25), \tfrac{1 - i/n}{25} \right\} & \texttt{slowdecay}.
% % \end{cases}
% % \begin{aligned}
% %     &\underbrace{\sigma_i := 100^{(i-1)/(n-1)}}_{\text{intro}}
% %     ,\quad
% %     &\underbrace{\sigma_i := 100^{(\lceil i/2 \rceil-1)/(n/2-1)}}_{\text{doubles}}
% %     ,
% %     \\
% %     &\underbrace{\sigma_i := \exp(-i/25)}_{\text{fastdecay}}
% %     ,\quad 
% %     &\underbrace{\sigma_i := \max\{\exp(-i/25),(1-i/n)/25  \}}_{\text{slowdecay}}.
% % \end{aligned}
% \end{equation}
\begin{align*}
\sigma_i &\coloneqq 0.01^{\tfrac{i-1}{\,n-1}};\tag{\texttt{intro}} \label{eq:intro} \\
\sigma_i &\coloneqq0.01^{\tfrac{\lceil i/2 \rceil - 1}{\,n/2 - 1}}; \tag{\texttt{doubles}} \label{eq:doubles} \\
\sigma_i &\coloneqq \exp(-i/25); \tag{\texttt{fastdecay}} \label{eq:fastdecay} \\
\sigma_i &\coloneqq \max\!\left\{ \exp(-i/25), \tfrac{1 - i/n}{25} \right\}.\tag{\texttt{slowdecay}} \label{eq:slowdecay}
\end{align*}
These spectra are visualized in \cref{fig:spectra}.
The \ref{eq:intro} example consists of geometrically spaced eigenvalues from $1/100$ to $1$. 
\Cref{fig:intro} uses this spectrum with target rank $k=200$.
The \ref{eq:doubles} example also consists of eigenvalues geometrically spaced from $1/100$ to $1$, but each eigenvalue has multiplicity 2. 
The \ref{eq:fastdecay} and \ref{eq:slowdecay} examples are from \cite{TW23}.

Because of the rotational invariance of the Gaussian starting block $\mat{G}$, the convergence behavior of the RBKI algorithm is invariant under rotation of the input matrix $\mat{A}$ and thus depends only on the singular values of $\mat{A}$.
To make timing results representative for general matrices, we do not use any sparse matrix libraries, treating and storing $\mat{A}$ as a dense \(n \times n\) matrix.
We implement \cref{alg:block_krylov_LRA} using an implementation of the block-Lanczos algorithm that performs full reorthogonalization at each iteration.
Our experimental results appear in \cref{fig:intro_ex,fig:doubles_ex,fig:fastdecay_ex,fig:slowdecay_ex}.
Our conclusions are much the same as \cref{fig:intro} in the introduction: Across all instances, the matrix-vector complexity of RBKI is smallest with block size $b = 1$, but the practical efficiency is optimized with the intermediate block size $b = 20$.
% Perhaps surprisingly, RBKI with block size $b = 1$ is effective 

\begin{figure}[t]
    \centering
    \includegraphics[scale=0.8]{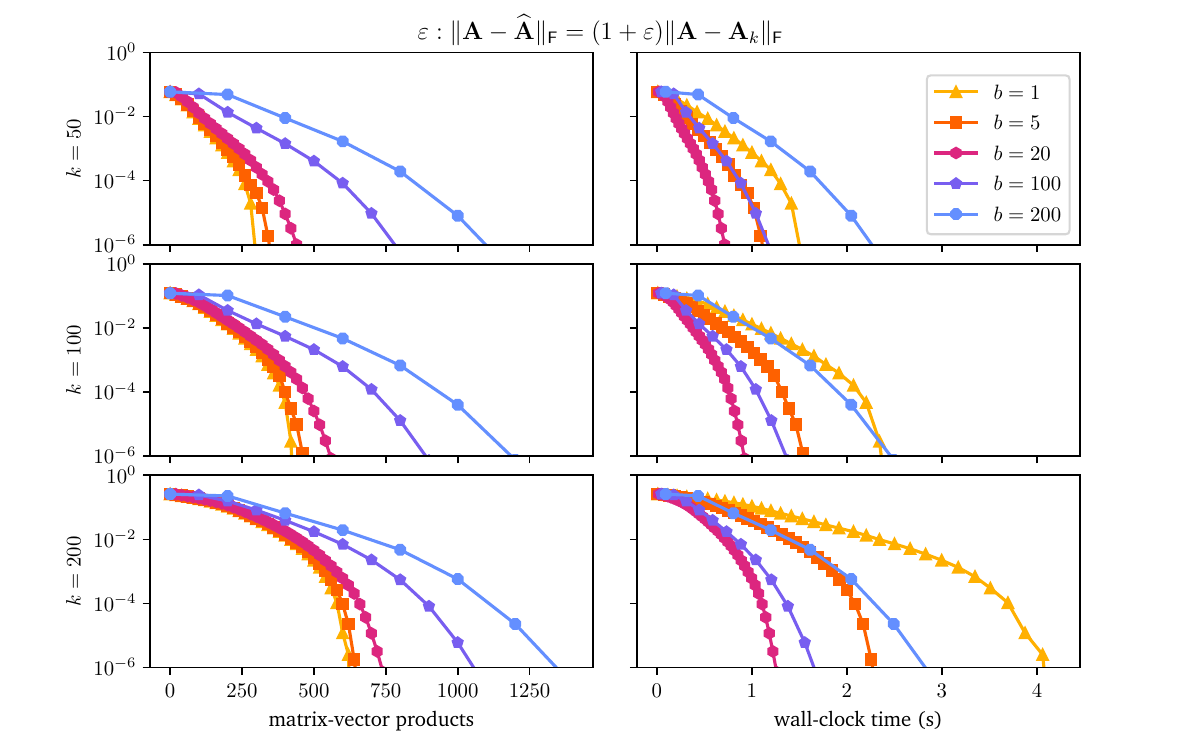}
    \caption{Convergence of \cref{alg:block_krylov_LRA} on \ref{eq:intro} problem for several values of $k$. }
    \label{fig:intro_ex}
\end{figure}

\begin{figure}[t]
    \centering
    \includegraphics[scale=0.8]{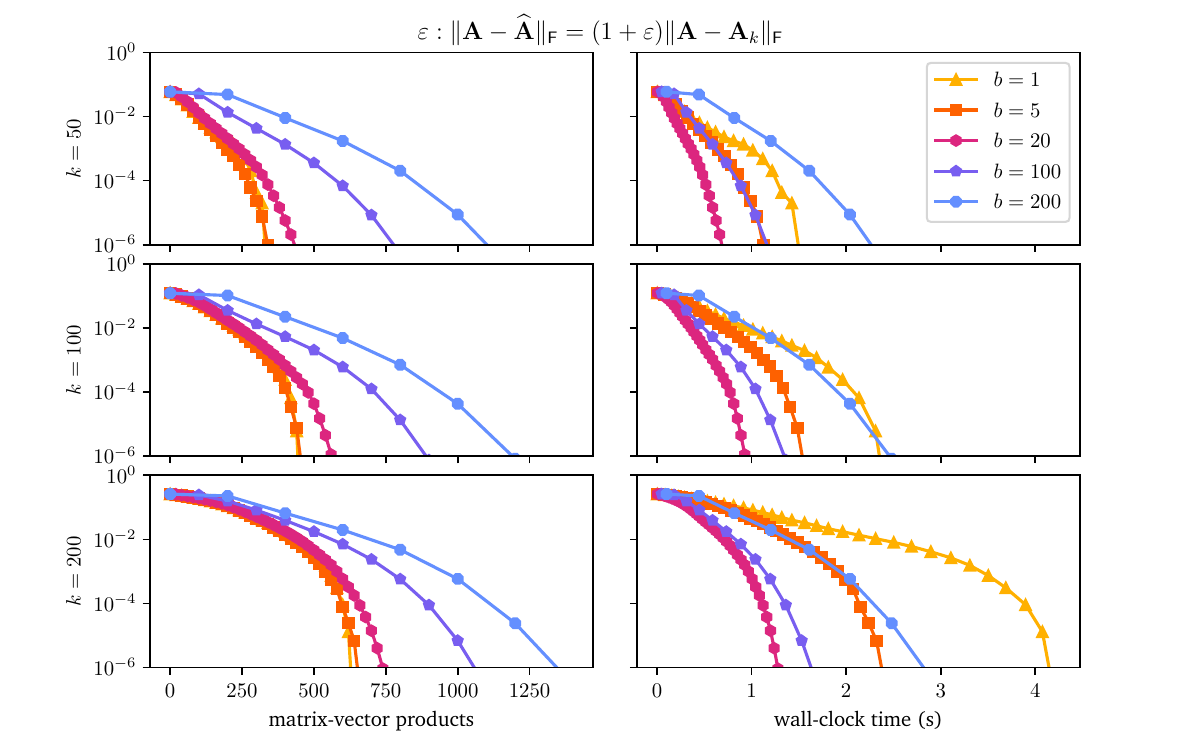}
    \caption{Convergence of \cref{alg:block_krylov_LRA} on \ref{eq:doubles} problem for several values of $k$.}
    \label{fig:doubles_ex}
\end{figure}

\begin{figure}[t]
    \centering
    \includegraphics[scale=0.8]{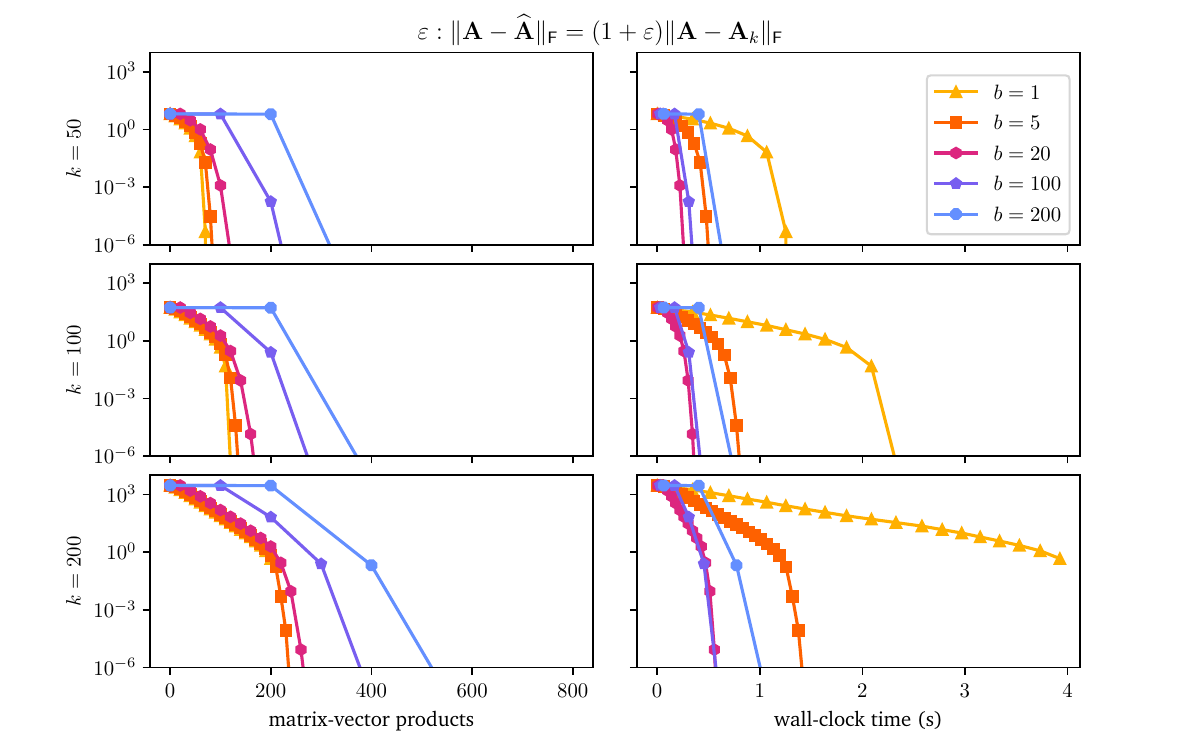}
    \caption{Convergence of \cref{alg:block_krylov_LRA} on \ref{eq:fastdecay} problem for several values of $k$.}
    \label{fig:fastdecay_ex}
\end{figure}

\begin{figure}[t]
    \centering
    \includegraphics[scale=0.8]{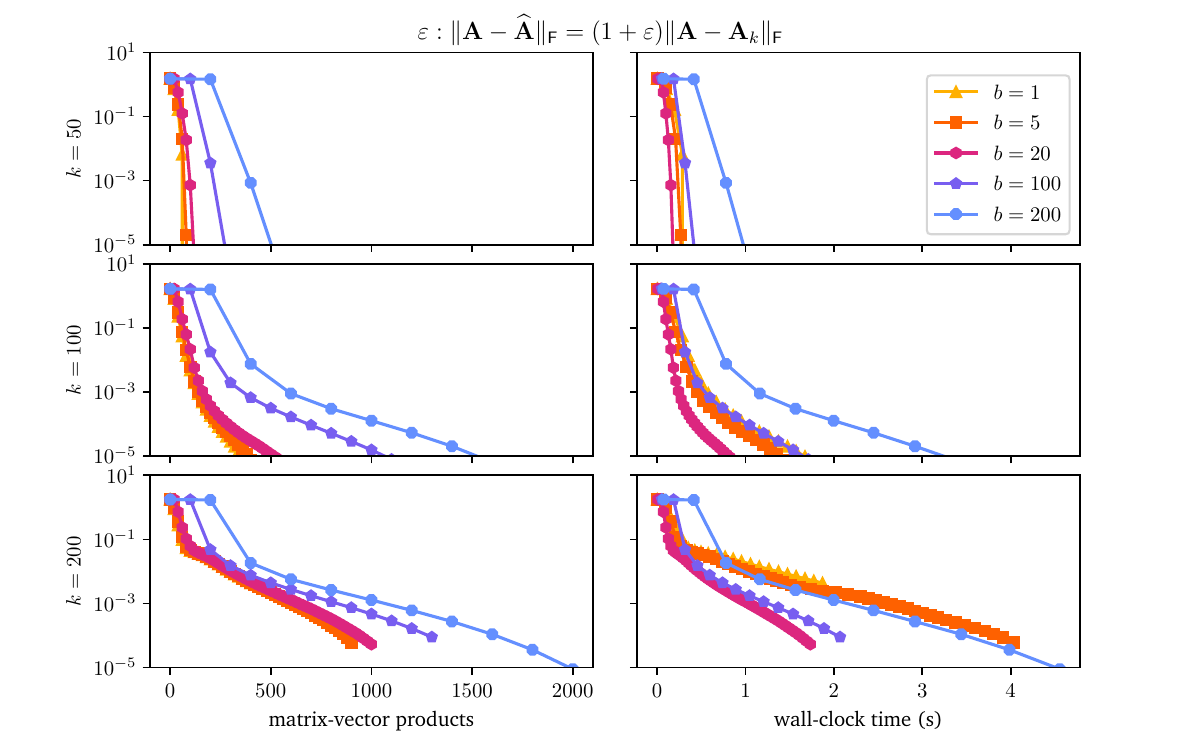}
    \caption{Convergence of \cref{alg:block_krylov_LRA} on \ref{eq:slowdecay} problem for several values of $k$.}
    \label{fig:slowdecay_ex}
\end{figure}